\documentclass[
    format=acmsmall, 
    anonymous=false,
    natbib=false,
    review=false,
    nonacm=true, 
]{acmart}

\usepackage{fancyvrb} 
\usepackage{rotating} 
\usepackage{makecell} 

\usepackage{xspace} 
\newcommand{\SELECT}{\lstinline'SELECT'\xspace}
\newcommand{\FROM}{\lstinline'FROM'\xspace}
\newcommand{\AND}{\lstinline'AND'\xspace}
\newcommand{\NULL}{\lstinline'NULL'\xspace}
\newcommand{\GROUPBY}{\lstinline'GROUP BY'\xspace}

\usepackage{fontawesome5} 
\newcommand{\accmat}{\href{https://github.com/laowantong/sqlab/blob/main/pub/2024-10. Learning SQL from Within/code.ipynb}{\textcolor{ACMDarkBlue}{\faFile*[regular]}}\xspace}

\usepackage{pifont} 
\newcommand{\cmark}{\ding{51}}
\newcommand{\gcmark}{\textcolor{lightgray}{\ding{51}}}

\usepackage{wrapfig} 

\usepackage{forloop}
\newcounter{x}\newcommand{\bullets}[1]{%
    \forloop{x}{0}{\value{x} < #1}{$\bullet$}%
    \forloop{x}{#1}{\value{x} < 4}{$\circ$}%
}

\usepackage{listings}
\makeatletter 
\newcolumntype{V}[1]{>{\topsep=0pt\@minipagetrue}p{#1}<{\vspace{-\baselineskip}}}
\makeatother

\newtheorem{definition}{Definition}

\patchcmd{\subequations}{}%
{}{}{}

\newcommand{\tablecaption}[1]{\vspace{\bigskipamount}\caption{#1}\vspace{-\bigskipamount}}

\RequirePackage[
  datamodel=acmdatamodel,
  style=acmnumeric,
  backend=biber,
  defernumbers=true
  ]{biblatex}
\acmJournal{TOCE}
\acmSubmissionID{}
\title{Learning SQL from within}
\subtitle{Integrating database exercises into the database itself}
\author{Aristide Grange}
 \orcid{0000-0001-5777-4603}
 \affiliation{%
    \institution{Université de Lorraine}
    \department{LCOMS EA7306}
    \city{Metz}
    \country{France}
}
\email{aristide.grange@univ-lorraine.fr}

\keywords{SQL, database, education, game-based learning, fingerprinting}
\ccsdesc[500]{Applied computing~Education}
\ccsdesc[500]{Information systems~Structured Query Language}
\ccsdesc[500]{Social and professional topics~Computing education programs}

\begin{abstract}
    SQL adventure builder (SQLab) is an open-source framework for creating SQL games that are embedded within the very database they query.
    Students' answers are evaluated using query fingerprinting, a novel technique that allows for better feedback than traditional SQL online judge systems.
    Fingerprints act as tokens that are used to unlock messages encrypted in an isolated auxiliary table. These messages may include hints, answer keys, examples, explanations, or narrative elements. They can also contain the problem statement of the next task, which turns them into nodes in a virtual DAG with queries as edges. This makes it possible to design a coherent adventure with a storyline of arbitrary complexity.
    
    This paper describes the theoretical underpinnings of SQLab's query fingerprinting model, its implementation challenges, and its potential to improve SQL education through game-based learning. The underlying concepts are fully cross-vendor, and support for SQLite, PostgreSQL and MySQL is already available. As a proof of concept, two games, 30 exercises and one mock exam were tested over a three-year period with about 300 students.
\end{abstract}

\begin{document}

\VerbatimFootnotes 
\maketitle

\section{Introduction}

Since its introduction in the 1970s by IBM researchers Donald D. Chamberlin and Raymond F. Boyce \cite{elmasri2017}, SQL has become an enormously complex language, with the 2023 standard spanning over 4,300 pages \cite{sql_standard_2023}. SQL betrays its age in its syntax, which is difficult to learn and difficult to use \cite{shute2024}. However, it remains the \emph{de facto} language for relational database management systems (RDBMSs), widely used in industry and widely taught in high schools and universities \cite{cscurricula2023,smith2022,tuparov2022}.

SQL educators must find effective ways to engage students in the learning process, help them understand the underlying principles of database design and querying, and provide them with appropriate feedback. Traditional approaches to teaching SQL typically rely on textbooks, lectures, and a series of exercises that are either checked manually by instructors or evaluated through automated systems. Textbooks provide theoretical knowledge, lectures offer clarification and deeper insight, and exercises allow for hands-on practice. However, instructors may struggle to provide timely and detailed feedback, especially in large classes. Feedback from automated systems is generally limited to a set of formulaic messages, leaving students to figure out their mistakes on their own. Moreover, these systems often compare the output of a student's query against an expected result table, a process known as execution matching. This can lead to false positives, where queries produce the correct results on a particular instance of the database but fail to generalize to all possible cases.

We introduce \textbf{SQL adventure builder} (\textbf{SQLab}), a novel framework for creating SQL games with the potential to be more engaging and more accurate in their feedback and assessment capabilities. They rely on cryptographic concepts like hashing (individual rows of a table), fingerprinting (SQL queries), and decrypting (feedback). As an interesting side effect, this approach allow us to integrate the exercises on the database into the database itself, turning into data what is generally thought of as code. Consequently, such games can be distributed as standalone database dumps, and played through any DBMS administration tool, eliminating the need for special-purpose software. As for the exercises themselves, they can range from a collection of unrelated questions, to an exam, to a full-blown tutorial, to a coherent adventure with a linear or DAG-shaped storyline.

SQLab has been developed and tested over three years with about 300 students in two database courses at \anon{Université de Lorraine, France}. It is published as open-source software at \cite{sqlab}. Several adventures are available or in preparation for three DBMSs: SQLite, PostgreSQL and MySQL, with more to come.


\subsubsection*{Terminology note}

CRUD (Create, Read, Update, Delete) operations on data form the core of any SQL instruction \cite{tuparov2022}, whether in Computer Science curricula \cite{cscurricula2023}, in professional training programs, or even in non-technical disciplines such as business \cite{smith2022}. Most of the time \cite{bowman2001, elmasri2017, taipalus2018, tuparov2022, mysql, oracle}, they are referred to as Data Manipulation Language.
However, in Part 2: “Foundation” of the SQL standard ISO/IEC 9075 \cite{sql_standard_2023}, section “14: Data Manipulation” does \emph{not} include the Read operation, which has its own section, “7: Query Expressions”.
Notable non-academic sources, such as Wikipedia (art. “DML” \cite{wiki:dml}) and PostgreSQL \cite{postgresql}, adhere to the terminology of the standard.
Since this distinction better aligns with our purposes, the present paper will also refer to data retrieval operations (SELECT) as \textbf{queries}, and data modification operations (INSERT, UPDATE, DELETE) as \textbf{DML statements}.


\subsubsection*{Reproducibility}

Throughout this paper, a clickable symbol \accmat is placed near each figure, table or example whose full code is provided in the accompanying notebook at \href{https://github.com/laowantong/sqlab/tree/main/pub}{\cite{sqlab}\texttt{/pub}}.

\subsubsection*{Plan}

The remainder of this paper is organized as follows.

Section~\ref{sec:related_works} (related works) reviews the literature on SQL error categorization and query matching, and proceeds to describe existing SQL games with respect to their design principles. 

Section~\ref{sec:playing} (playing an SQLab adventure) presents an SQLab game from the student's perspective. It demonstrates the practical utility of query fingerprinting, the core concept of our framework.

Section~\ref{sec:fingerprinting} (fingerprinting) begins with a thorough theoretical analysis of fingerprint: how it is calculated, what categories of queries it applies to, and how accurately it can differentiate them. We define \emph{starring}, a query transformation such that, under certain conditions, fingerprint equality of the original queries is equivalent to execution matching of the starred queries. We then discuss in detail the implementation challenges of fingerprinting formulas, with a focus on aggregation functions: the distribution of their outcomes, the properties of their composition, and their availability across DBMSs.

Section~\ref{sec:life_cycle} (life cycle of an SQLab database), as opposed to Section~\ref{sec:playing}, describes an SQLab game from the designer's perspective: what source material is needed, how the scenario is layed out and the tasks are encoded, what happens during the build process, and what to expect during the practice phase.

Section~\ref{sec:conclusion} (conclusion) summarizes the main contributions of this paper and outlines future directions for research.

\section{Related works} \label{sec:related_works}

\subsection{SQL errors}

The literature on this topic originates from the late 1970s, with early references provided by Ahadi et al. \cite{ahadi2016}. In the same paper, the authors present the first large scale analysis of SQL \textbf{syntax errors} and demonstrate that 54\% of student attempts present such mistakes.

Another 40\% of the student queries are syntactically correct, but nevertheless return incorrect results for at least one instance of the database. Brass and Golberg \cite{brass2006} provide a comprehensive list of the subset of these errors that can be identified without knowledge of the intended task. Taipalus et al. \cite{taipalus2018} refer to these as \textbf{semantic errors}, as opposed to \textbf{logical errors}, which require knowledge of the task to be identified. In turn, they put forth an exhaustive list of these logical errors.

Correct queries, that return the expected result on any instance of the database, represent the remaining 6\% of the attempts analyzed in \cite{ahadi2016}. Of these, some are considered \textbf{gold queries}, while others could be formulated in a simpler fashion, and are thus classified as \textbf{complications}. Like semantic errors, complications can be identified without knowledge of the task.
This categorization is conveniently summarized in Table~\ref{tab:taipalus}, taken from \cite{taipalus2018}.

\begin{table}[h!]
    \tablecaption{Taipalus et al.'s error classes. In bold, errors evident without knowledge of the intended task. Note that Brass and Goldberg \cite{brass2006} call \emph{semantic errors} what is divided here into \emph{logical errors} and \emph{semantic errors}. Moreover, they study \emph{possible runtime errors}, which are left outside this classification.}
    \begin{tabular}{ccccc}
        \multicolumn{2}{c}{\em Correct result table} 
        & \multicolumn{2}{c}{\em Incorrect result table}
        & \multicolumn{1}{c}{\em No result table}
        \\
        \cmidrule(lr){1-2} \cmidrule(lr){3-4} \cmidrule(lr){5-5}
        gold & \textbf{complication} & logical error & \textbf{semantic error} & \textbf{syntax error}
    \end{tabular}
    \label{tab:taipalus}
\end{table}

Some results are transversal with respect to these categories. For instance, Taipalus and Perälä show that logical errors and complications are the most \textbf{persistent errors}, namely, those that may recur even after multiple attempts to correct them \cite{taipalus2019}.
They hypothesize that these errors are often rooted in fundamental misconceptions or misunderstandings about SQL. Miedema et al. extend significantly this investigation in \cite{miedema2024,miedema2022b,miedema2022a}. 

\subsection{Assessment of the correctness of SQL queries}

Two queries are said to be \textbf{semantically equivalent} if they return the same result table on every possible database instance. Although the query equivalence problem is in general undecidable \cite{chu2017a}, it has numerous real-world applications, including
query optimization,
query caching \cite{zhou2019},
text-to-SQL translation \cite{kim2020,kanburoglu2024,zhong2020,zhao2024},
grading or providing feedback to students \cite{burgstaller2023,kaur2023,koberlein2024,nayak2024,wang2020,yang2022}, etc. There are two main methods for assessing the equivalence between a submitted SQL query and a designated gold query:
\textbf{execution matching} compares the result tables, at the risk of false positives;
\textbf{code matching} compares the queries themselves, usually at the risk of false negatives.

Among the games we will review in Section~\ref{sec:sql_games}, SQL Island, Select Star SQL, and GalaXQL only rely on execution matching, while aDBenture combines execution and code matching. As for SQLab, it compares the result tables of a “starred” version of the original queries, a concept we will introduce in Section~\ref{sec:theory}.

\subsubsection{Execution matching} \label{sec:mutation}
This approach is at the heart of every Online Judge System. An \textbf{OJS}, by definition \cite{wasik2018}, “assesses solutions [...] based on a set of test instances”\,---\,which, in the context of SQL queries, is one or several database instances. However, as observed in \cite{wang2024}, OJS “only compare the execution results of queries with the correct results, without assessing the correctness of the query logic”. As an unfortunate consequence, “to deceive OJS and obtain correct evaluations, students may deliberately [...] construct queries that execute correctly on specific test cases but are logically incorrect”. This gave rise to “a new type of logical error called \textbf{deceptive error}, which occurs when students construct queries that pass specific test cases but fail to solve the actual problem”.

To overcome the inability of OJS to assess the logic of an SQL query, the designer “should give special attention to both the type of the problem and the preparation of [the dataset]” \cite{wasik2018}. The latter requirement can be partially fulfilled by means of mutation testing, a technique borrowed from software engineering \cite{jia2011}. A \textbf{mutation} is defined in \cite{chandra2015} as "a single (syntactically correct) change of the original query", while a \textbf{mutant} is "the result of one or more mutations on the original query". Mutation has recently been used to "repair" the slightly incorrect queries sometimes emitted by text-to-SQL tools \cite{yang2023}. However, in the majority of scenarios, it is applied on correct queries to generate incorrect ones, the end goal being to assess the quality of a test suite \cite{chan2005, tuya2007} and/or to generate new test data \cite{shah2011,suarezcabal2017,chandra2015,zhong2020}. A database state is said to \textbf{kill} a mutant if executing both the mutant query and the gold query yields different tables. An \textbf{equivalent mutant} is a mutant which no database state can kill \cite{tuya2007}.

Mutations of interest are those that can be systematically generated. For example, the tool SQLMutation \cite{tuya2006} transforms a given query by applying four mutation operators, categorized according to their effect on: the main SQL clauses (SC); the operators within conditions and expressions (OR); the handling of NULL values (NL); the identifiers, including column references, constants, and parameters (IR). XDATA \cite{chandra2015} explores a mutation space involving join type, selection predicates, and unconstrained aggregations.

\subsubsection{Code matching}
Assessing the equivalence of SQL queries by static analysis can be carried out in a variety of ways.
The simplest approach, \emph{exact string matching}, is too strict to be useful in practice. \emph{Exact set matching} \cite{yu2018,kaur2023} is more lenient, as it only requires the student to use the same keywords and identifiers in the same clauses as the gold query. Such text-based comparisons have the advantage of being applicable to syntactically incorrect queries, but they may suffer from a lack of robustness to minor variations in the queries.
More sophisticated approaches entail parsing (syntactically correct) queries into formalizations amenable to semantic comparison,
such as
abstract syntax trees \cite{wang2020,koberlein2024},
tree-like structures of keywords and identifiers \cite{yang2022},
symbolic representations \cite{zhou2019}, 
algebraic expressions \cite{chu2017a}, 
functions over cardinal numbers or univalent types \cite{chu2017b}. 

\subsection{SQL games} \label{sec:sql_games}

The concept of integrating education and entertainment can be traced back to at least the time of Jerome\footnote{
    “Get for her a set of letters made of boxwood or of ivory and called each by its proper name. Let her play with these, so that even her play may teach her something.”, Ep. CVII to Laeta (c. 400) \cite{schaff_1886}.
}. However, it was not until 1970 that the term “serious game” was first coined \cite{abt1970,djaouti2011}. Then, in the early 2000s, the rapidly developing and sometimes-maligned \cite{funk2004} industry of video games gained unexpected advocates among academic researchers. Prensky \cite{prensky2001,prensky2003} posited that video games align with how a “digital native” learns\,---\,through experimentation, problem-solving, and immediate feedback.
Gee \cite{gee2003} added that they can teach valuable literacy and learning strategies by immersing players in complex systems, and encouraging them to adapt and think critically. The advent of “game-based learning” (GBL) gave rise to the development of numerous educational games, which in turn prompted a multitude of studies on their impact on student's engagement and learning outcomes (see \cite{%
    bond2020,
    boyle2016,
    connolly2012,
    egenfeldt2006,
    ekin2023,
    papastergiou2009
} for literature reviews and meta-analyses).

A small subset of these games involves SQL learning and training. A Google Scholar search for “SQL game” returns over 50,000 results, but most of them are either poor fits, duplicates, gamified tutorials (e.g. \cite{ward2015,pustulka2021}), SQL-injection cybersecurity games, or sequences of problems lacking a storyline. Non-academic sources provide a few more examples. All in all, we identified four web-based text adventures \cite{batista2019,burgstaller2023,canale2022,schildgen2015}, two augmented reality (AR) games \cite{lupano2021,tuparov2022}, one text-mode massively multiplayer online role-playing game (MMORPG) \cite{schemaverse_repo}, and even one 3D role-playing game developed as a mod of NeverWinter Nights \cite{soflano2011,soflano2015}. We will discuss just three of these \cite{burgstaller2023,canale2022,schildgen2015}, representative of the others in terms of the orthogonal properties we are interested in. For comparison sake, we also include two quasi-games \cite{galaxql_online,select_star_sql_online}, along with our own proposed solution \cite{sqlab}.

\subsubsection{Properties studied}

Table~\ref{tab:game_properties} orders this selection by inception date, and introduces four terms that will be defined in the following paragraphs.

\begin{table}[h!]
    \tablecaption{Characteristics of some SQL games. The main novelty of SQLab is that it allows the creation of games that are both \emph{directed} and \emph{standalone}, with highly \emph{specific feedback}.}
    \begin{tabular}{llcccccc}
  year
  & name
  & \shortstack{directed\\($\neg$ undirected)}
  & \shortstack{standalone\\($\neg$ embedded)}
  & \shortstack{stateful\\($\neg$ stateless)}
  & \shortstack{feedback\\level}
  \\
  \hline
  2005 & GalaXQL \cite{galaxql_online} &\cmark &  ~  & \cmark & \bullets{4} (4) \\
  2014 & SQL Island \cite{schildgen2015,sql_island_online}  &\cmark &  ~  & \cmark & \bullets{2} (2) \\
  2018 & SQL Murder Mystery \cite{canale2022,sql_murder_mystery_repo} & ~ & \cmark  & ~ & \bullets{0} (0) \\
  2018 & Select Star SQL \cite{select_star_sql_online} &\cmark &  ~  & ~ & \bullets{1} (1) \\
  2023 & aDBenture games \cite{adbenture_online,burgstaller2023} &\cmark &  ~  & ~ & \bullets{3} (3) \\
  2024 & SQLab games \cite{sqlab_repo} & \cmark & \cmark  & \cmark & \bullets{4} (4) \\
\end{tabular}

    \label{tab:game_properties}
\end{table}

\begin{enumerate}
\item We say that an SQL game is \textbf{directed} if the player is guided through one or more sequences of tasks predefined by the instructor. It is said \textbf{undirected}\footnote{
    Most authors \cite{burgstaller2023, hattie2008, xinogalos2022} describe such games as \emph{self-directed}, but we prefer \emph{undirected} for its straightforward antonym \emph{directed}.
} if players are free to craft their own strategies to solve the tasks. The latter may better reflect actual usage of SQL in a business setting, but it prevents the adventure from being used as a tutorial, making it less suitable for learning than for training.

\item An SQL game is \textbf{standalone} if it can be distributed as a self-contained database dump. Being a “normal” database, it is potentially playable through the full range of browser-based and desktop administration tools, either dedicated (e.g., phpMyAdmin \cite{phpmyadmin} for MySQL \cite{mysql} and MariaDB \cite{mariadb}, etc.) or generic (e.g., DataGrip \cite{datagrip}, DBeaver \cite{dbeaver}, SQuirreL SQL Client \cite{squirrel}, etc.), as well as through vendor-specific command-line tools (e.g., \texttt{mysql}, \texttt{psql}, \texttt{sqlite3}, etc.). This flexibility allows instructors to familiarize their students with both SQL and a “real-life” tool, rather than requiring them to learn a purpose-built interface they won't use afterward. Additionally, individual learners may select the environment that best aligns with their workflow, comfort level, or special needs.

\item An SQL game is either \textbf{stateful} or \textbf{stateless} depending on whether the player is expected to modify the state of the underlying database (with DML statements such as UPDATE, DELETE, INSERT), or simply retrieve data (with SELECT queries). This distinction has significant implications for the design of the game. Statefulness requires each student to work on a fully-functional sandbox database to avoid altering the game state of their peers. This can be set up locally on their own machine (via a full installation of the DBMS, a virtual machine, or a Docker container), hosted on a generic cloud platform, or provided at university level \cite{tuparov2022}. In contrast, statelessness allows all students to share the same database instance, but precludes the use of numerous interesting plot devices. 

\item Apart from the DBMS' error messages, we distinguish five levels in the \textbf{feedback} delivered to the players when they submit a query: 0 for \textbf{no feedback}, 1 for \textbf{binary} (“Correct” / “Incorrect”), 2 for \textbf{generic} (e.g., “Too many rows”), 3 for \textbf{templated} (e.g., “Missing column \{x\}”), and 4 for \textbf{specific} (e.g., “Filter out the evil inhabitants, not the friendly ones”).

\end{enumerate}

It bears noting that the presence of any one of the four aforementioned characteristics is inherently more demanding than its absence: a system that enables the creation of directed and stateful games \emph{a fortiori} supports undirected and stateless ones; a standalone game can obviously be embedded in a website (see Section~\ref{sec:embedded-version} for a discussion on an embedded version of SQLab games); and a feedback message pertaining to level $i+1$ may be rephrased at level $i$. However, none of these implications is reciprocal.

\subsubsection{SQL Island}

This online tutorial from Technical University of Kaiserslautern stands as the leading example of the \emph{directed} / \emph{embedded} model: a dedicated web application \cite{sql_island_online} tightly integrated with its underlying database and scripting every aspect of storyline and interaction \cite{sql_island_repo}. The player assumes the role of a character stranded on an island after a plane crash. To escape, they must retrieve or \emph{modify} various data, the latter meaning that the game is \emph{stateful}\footnote{
    SQL Island has sometimes \cite{xinogalos2022, burgstaller2023} been referred to as \emph{stateless} for its inability to resume a game without starting over from the beginning. This property, which we prefer to call \textbf{resumability}, is not intrinsic but rather an implementation detail. As a matter of fact, SQL Island has recently become resumable without any change to its gameplay \cite{sql_island_repo}.
}. DML statements enable several classical devices of the text-adventure genre: gathering resources, selling them, earning a salary, killing the villains, and rescuing their victim. SQLite is used as the back-end because its file-based nature makes it easy to create a new database instance for each game session \cite{schildgen2015}. Currently, SQL Island only provides a \emph{generic feedback} (level 2 of 4) on the number and content of the resulting rows and columns, without offering further analysis or guidance \href{https://github.com/jschildgen/sql-island/blob/master/lang_en.csv}{\cite{sql_island_repo}\texttt{/blob/master/lang\_en.csv}}.

Note that, as a proof-of-concept of our game-engine capabilities, we have developed and open-sourced three adaptations of SQL Island (for SQLite, MySQL and PostgreSQL). SQLab Island \cite{sqlab_island} is still \emph{directed} and \emph{stateful}, but gains the \emph{standalone} and \emph{specific feedback} properties. Increasing the feedback level was often cited among the improvements proposed by the players of SQL Island \cite{xinogalos2022}.

\subsubsection{SQL Murder Mystery}

Published by Northwestern University's Knight Lab, this detective story \cite{sql_murder_mystery_repo} serves as a rare example of an \emph{undirected} / \emph{standalone} SQL training game, making it conceptually the opposite of SQL Island.
As a testament of its undirected nature, solutions found online range from 4 to 15 steps, differing significantly in style, elegance, and accuracy: for instance, at the beginning of the game, asked to find which “witness lives at the last house on Northwestern Dr”, some use a nested \texttt{SELECT max(...)} statement, others an \texttt{ORDER BY} with (or even without) a \texttt{LIMIT} clause.

Being standalone, SQL Murder Mystery is more of a database than a software application. It comes as a downloadable SQLite dump, and an initial prompt to kickstart the investigation\footnote{
    Due to popular demand, it was later wrapped in an (almost generic) web application and accompanied by a static tutorial \cite{sql_murder_mystery_online}, allowing beginners and non-technical users to play without downloading or installing anything.
}. Giving the players full access to the data raises an interesting question: how can they be prevented from stumbling upon the solutions, either intentionally or accidentally? Before presenting the authors' answer, we need to walk through the game. Its story unfolds in just six narrative steps. The player retrieves the first clue from the \texttt{description} column in the \texttt{crime\_scene\_report} table, filtered by the crime's \texttt{date}, \texttt{location}, and \texttt{type}. The next two clues are gathered from the \texttt{transcript} column in the \texttt{interview} table using the witness IDs deduced from the report. These clues point to a suspect whose name, when inserted into the \texttt{solution} table, triggers a fourth message. The player then interviews the murderer (fifth message) to trace the mastermind, leading to the sixth and final message.

In order to prevent deceptive errors (or cheating, in layman's terms), two distinct measures have been implemented.
The first one is reminiscent of Rivest's \textbf{chaffing} \cite{rivest1998}, i.e., “adding fake [rows] with [...] reasonable message contents, [...] intermingled with the good (wheat) [rows]”. Indeed, the \texttt{crime\_scene\_report} and \texttt{interview} tables contain only 4 “wheat” rows among 6,215 “chaff” ones. The fake rows have plausible but random data, including sentences from \emph{Alice in Wonderland} and various internet sources. While most of these filler messages make no sense in the context of the story, there is still a risk that an incorrect query might yield chaff that would be mistaken for an actual interview transcript, further confusing a struggling student: we call this the \textbf{misleading message issue}.
The second measure involves storing both suspect names and messages in an SQL trigger. The suspect names are hex-encoded before being compared against two hard-coded values, and the messages themselves don't contain any revealing information. This simple \textbf{obfuscation} technique ensures that the answers cannot be readily seen in the trigger's code. Faced with the same problem, we resorted to a combination of hashing, salting and encryption, which is more complex to implement, but provably more robust, and does not suffer from the misleading message issue.

SQL Murder Mystery has \emph{no feedback} mechanism (level 0 of 4), except for the trigger that updates the table \texttt{solution} with a congratulatory message \cite{sql_murder_mystery_repo}.

\subsubsection{aDBenture}

Recently developed by the University of Klagenfurt, this framework allows for the creation of \emph{directed} / \emph{embedded} / \emph{stateless} SQL games \cite{adbenture_online}. It relies on a dedicated learning environment, as generic administration tools “do not provide enough feedback” (this was indeed the primary challenge we aimed to address in our own design). The authors highlight that their solution, in contrast with all others, offers the teachers a backend to analyze student work. However, from our point of view, this is not an inherent aspect of the model: once query logging has been enabled at the administrative level, it is always possible to analyze the logs with a separate program (for instance, SQLab provides a sub-command \texttt{report} for this purpose).

In an “aDBenture”, the player's query is first checked against a pool of known gold queries; if it matches one exactly, the task is marked correct; otherwise, if its result table is correct, it is flagged for manual review and potential addition to the pool. Any incorrect answer enters a pipeline of comparisons against both the result table and the code of the query; a feedback message is generated for the first issue encountered, whether it is a data discrepancy, or a missing or unnecessary column, alias, table, WHERE condition, or sorting criterion \cite{burgstaller2023}. The teacher can set the feedback level to be either \emph{binary} (level 1 of 4) or \emph{templated} (3 of 4).

\subsubsection{Schemaverse}

This unique space MMORPG runs entirely within a PostgreSQL database. Players issue SQL commands to manage their in-game assets, such as creating ships, moving them across the game universe, or attacking other players' fleets. However, as of 2024, the repository \cite{schemaverse_repo} appears to be inactive, and the central server \cite{schemaverse_online} is unavailable. Due to the difficulty to test it, it was not included in Table~\ref{tab:game_properties}.

\subsubsection{Other SQL “games”}
We put \emph{games} in quotation marks because the \emph{directed} / \emph{embedded} tutorials we briefly discuss here lack the playful aspect of the previous examples.

As a matter of fact, Select Star SQL \cite{select_star_sql_online} is based on a dataset documenting Texas death row inmates executed from 1976 to the present.
The feedback is \emph{binary} (level 1 of 4). The correctness is assessed by execution matching: thus, complications such as useless or extraneous elements go undetected. Inspection of the JavaScript code \cite{select_star_sql_repo} confirms that errors are identified solely by ensuring that the symmetric difference set is empty\footnote{
    \href{https://github.com/zichongkao/selectstarsql/blob/3fbda4d40cc42007dfea31600222679764dd197c/scripts/main.js#L188}{\cite{select_star_sql_repo}\texttt{/mysql/blob/3fbda4d40cc42007dfea31600222679764dd197c/scripts/main.js\#L188}.}
}.

In GalaXQL \cite{galaxql_online}, the user writes queries to manipulate a 3D galaxy. The feedback is \emph{specific} to each query (level 4 of 4). In the C++ wrapper \cite{galxql_repo}, each task has its own hardcoded set of tests and messages. We will see later how SQLab allows a similar level of guidance without requiring the designer to write any host-language code.


\section{Playing an SQLab adventure} \label{sec:playing}

\subsection{Setup}

An SQLab database is distributed as a self-contained SQL dump file. Executing this file creates and/or populates all necessary tables, functions and triggers. Students can download it and install the appropriate DBMS (currently, either PostgreSQL, MySQL, or SQLite\footnote{%
    Priority was given to these DBMSs because they are free, open-source and widespread, but nothing prevents other systems from being added in the future. Each version requires writing approximately 100 lines of Python, and a similar amount of SQL code in the target dialect.
}) on their own machine, or simply connect to the administration interface of the university server which hosts the database.

\subsection{Individual tasks} \label{sec:individual_tasks}

\subsubsection{One-pass fingerprinting}

A task is made up of at least two elements (Fig.~\ref{ex:fingerprinting_1}, above dotted line): a \textbf{problem statement} and a \textbf{token formula}. The token formula is an SQL expression that students have to copy and paste into the SELECT clause of their query (Fig.~\ref{ex:fingerprinting_1}, below dotted line). It calculates and distributes along a \texttt{token} column what can be considered as a fingerprint of an intermediate result of the query execution (more on this later). Then, the students call a predefined function \texttt{decrypt()} with this token as argument (for Fig.~\ref{ex:fingerprinting_1}, \lstinline|SELECT decrypt(279951804330132)|). Behind the scene, this function queries a single-column table of encrypted messages, \texttt{sqlab\_msg}, filters out all rows that fail to decrypt, and returns a single-cell result table with the decrypted feedback. This can be either a success message, or a specific hint to help the student correct their query. Whenever all rows are filtered out, a generic fallback error message is displayed.

\begin{figure}[!ht]
    \input{floats/fingerprinting_1}
    \caption{An SQLab exercise with a possible answer and its result table. The given formula is meant to be copied and pasted into the outer SELECT clause of the query. The resulting token can be directly used to unlock the correction (one-pass fingerprinting). \hfill\accmat}
    \label{ex:fingerprinting_1}
    \Description{Text.}
\end{figure}

\subsubsection{Two-pass fingerprinting}

Certain types of queries require one more copy-paste operation (Fig.~\ref{ex:fingerprinting_2}). In these cases, the formula contains a placeholder \texttt{(0.0)} that the students are instructed to replace by a control value extracted from the result table or, less frequently, from the query itself. Executing the updated query will then give the token its final value (not shown on the figure) to be used in the \texttt{decrypt()} function. This process is called \textbf{two-pass fingerprinting}. Although more complex, it is only required for more advanced queries, and then should only come into play once students are familiar with the basic game mechanics.

\begin{figure}[!ht]
    \input{floats/fingerprinting_2}
    \caption{A query with two-pass fingerprinting. Students are invited to replace the \textcolor{ACMOrange}{\texttt{(0.0)}} placeholder in the formula by the first number of the \texttt{employees} column (\textcolor{ACMOrange}{\texttt{1}}). \hfill\accmat}
    \label{ex:fingerprinting_2}
    \Description{Text.}
\end{figure}

\subsection{Progression through tasks}

Fig.~\ref{fig:activity_map} shows part of the structure of a collection of more than 500 feedback messages we created for a series of SQL practical sessions on a home-made French-language database called \emph{Sessform} \cite{sqlab_sessform}. It provides 30 basic exercises, 10 mock exam questions, and an adventure spanning no less than 36 episodes.

\begin{figure}[!ht]
    \fcolorbox{lightgray}{white}{\includegraphics[width=0.6\linewidth,clip,trim=32cm 11cm 1.8cm 21cm]{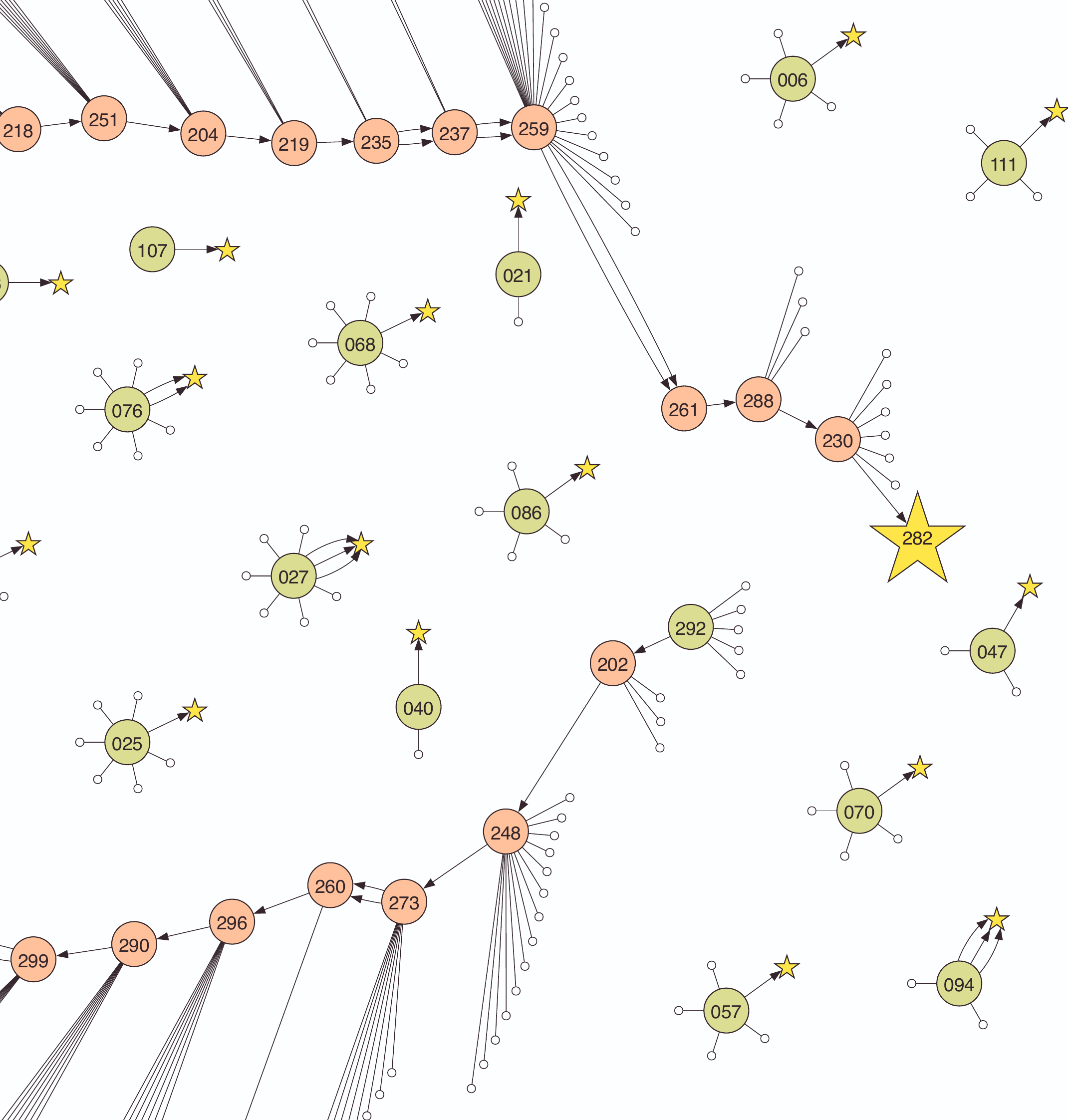}}
    \caption{
        Fragment of an activity map on SQLab Sessform (MySQL version).
        A green node is an entry point, i.e., a task whose access token is part of the statement given to the students: it can be either an independant exercise (e.g., 027, 040, 086) or the first episode of an adventure (292).
        A red node (e.g., 202, 248) is an intermediate point, i.e., a task whose access token is produced by a correct answer to its predecessor; multiple arcs (e.g., the arcs outgoing from 027 or 273) indicate that a task admits multiple solutions.
        A small blank node represents a specific hint-message triggered by a wrong answer (arrows omitted for clarity).
        A yellow star (e.g., 282) is an exit point, i.e., a message that congratulates the student for solving an exercise or completing an adventure.
        Entire map on \href{https://github.com/laowantong/sqlab_sessform/blob/main/mysql/activity_map.pdf}{\cite{sqlab_sessform}\texttt{/mysql/activity\_map.pdf}}.
        }
    \Description{A directed acyclic multigraph representing an adventure and a collection of SQL exercises.}
    \label{fig:activity_map}
\end{figure}

More generally, the underlying structure of the message table is a multigraph of one or more \textbf{tasks} of two possible types:

\begin{itemize}
\item \textbf{Exercises} are standalone tasks that can be tackled in any order. Topologically, they are star multigraphs: their internal node is the question-message, the arcs are the predicted answers, and the external nodes are the \emph{non sequitur} feedback-messages. At least one predicted answer is correct; its target node consists of a congratulatory message, the gold answer, zero, one or more variants, and optionally a textual fragment. There may be zero, one or more predicted wrong answers; each one targets a specific \textbf{hint}.
\item An \textbf{adventure} is a series of tasks structured by a narrative and/or didactic plot. Topologically, it is a directed acyclic multigraph with at least one source (the first task) and one sink (the epilogue). What we have said about the elements of an exercise still holds, except for the following points: there are \emph{several} internal nodes, called \textbf{episodes}; the question-messages usually include a fragment of the storyline and all the sections (congratulatory message, etc.) pertaining to the target of a correct answer.
\end{itemize}

Each task has two unique identifiers. First, as we have seen, its \textbf{access token}. Second, its \textbf{number}, a three-digit natural integer that is used as a suffix for both the title of the task (e.g., “Exercise [042]” in Fig.~\ref{ex:fingerprinting_1}) and the name of the salt function involved in the token calculation formula (\texttt{salt\_042()} in the same figure); somewhat counterintuitively, the numbering is not necessarily ordered: this allows the instructor to insert new tasks without having to renumber the existing ones. Moreover, if a task is an entry point (green nodes in Fig.~\ref{fig:activity_map}), the token is conveniently equal to the number of the task; for instance, the token of the “292nd” (actually first) episode of Sessform adventure is also 292, so the students can simply unlock the corresponding message by typing \lstinline|SELECT decrypt(292)|.

Players are encouraged to copy and paste into a separate text file all the success messages they receive (see Fig.~\ref{fig:pma} for an example). Each message includes a detailed correction (essential for review) and, in an adventure, a textual context (useful for catching missed details of a narrative or a tutorial), the next task, and the access token. The latter could later be used to resume a stateless adventure from that checkpoint. This primitive save mechanism is akin to the “game codes” which were prevalent before the advent of battery-powered memory and file-storage systems, and that the players would handwrite on a piece of paper to avoid losing their progress \cite{tobin2016}.

\section{Fingerprinting an SQL query} \label{sec:fingerprinting}

The calculation of access tokens is the cornerstone of our model. It allows us to determine whether a student's query is equivalent to a predicted query, and if so, to provide them with the appropriate feedback. This section describes the logical principles and the implementation of this mechanism.

\subsection{Concepts} \label{sec:fingerprinting_concepts}

As defined in \cite{broder1993}, a \textbf{fingerprint} is a short tag for larger objects such that, if two fingerprints are different, the corresponding objects are also certainly different. The goal is to avoid collisions (i.e., the same fingerprint for two different objects) as much as possible. In our case, an object is an SQL query on a given database. A fingerprint (or access token if we adopt the student's perspective) is a big integer calculated from the hashes of the rows involved in the top level part of this statement.

These hashes are not calculated during the execution of a query, but rather pre-calculated during the creation of the database, and stored in a supplementary column called \texttt{hash}. To ensure cross-table uniqueness, the required \textbf{hash function} $h$ takes as input not only the content of a row, but also the name of the table it belongs to.

Fingerprinting a query heavily relies on \emph{aggregation functions}, for their “ability to summarize information” \cite{vanrenesse2003}. To paraphrase \cite{Jesus2015}, an \textbf{aggregation function} $\mathcal{A}: \mathbb{N}^I \to O$ maps a multiset\footnote{$\mathbb{N}^I$ represents the set of all functions from the set $I$ to the set $\mathbb{N}$ of natural numbers. In other words, it describes how many times each element of $I$ appears in the multiset.} of elements from a domain $I$ to a value in a domain $O$.

Note that the input being a multiset implies that: 1) the order of the elements is irrelevant; 2) a given element may occur several times. Similarly, we can already rule out two classes of SQL (so-called) aggregation functions: 1) the order-sensitive ones (like \texttt{group\_concat()}, \texttt{string\_agg(), \texttt{first()}}, \texttt{json\_arrayagg()}); 2) the duplicate-insensitive ones (like \texttt{count(distinct)}).

\subsection{Algorithm}

\subsubsection{Overview}

\begin{figure}[!ht]
    \includegraphics[width=0.7\linewidth]{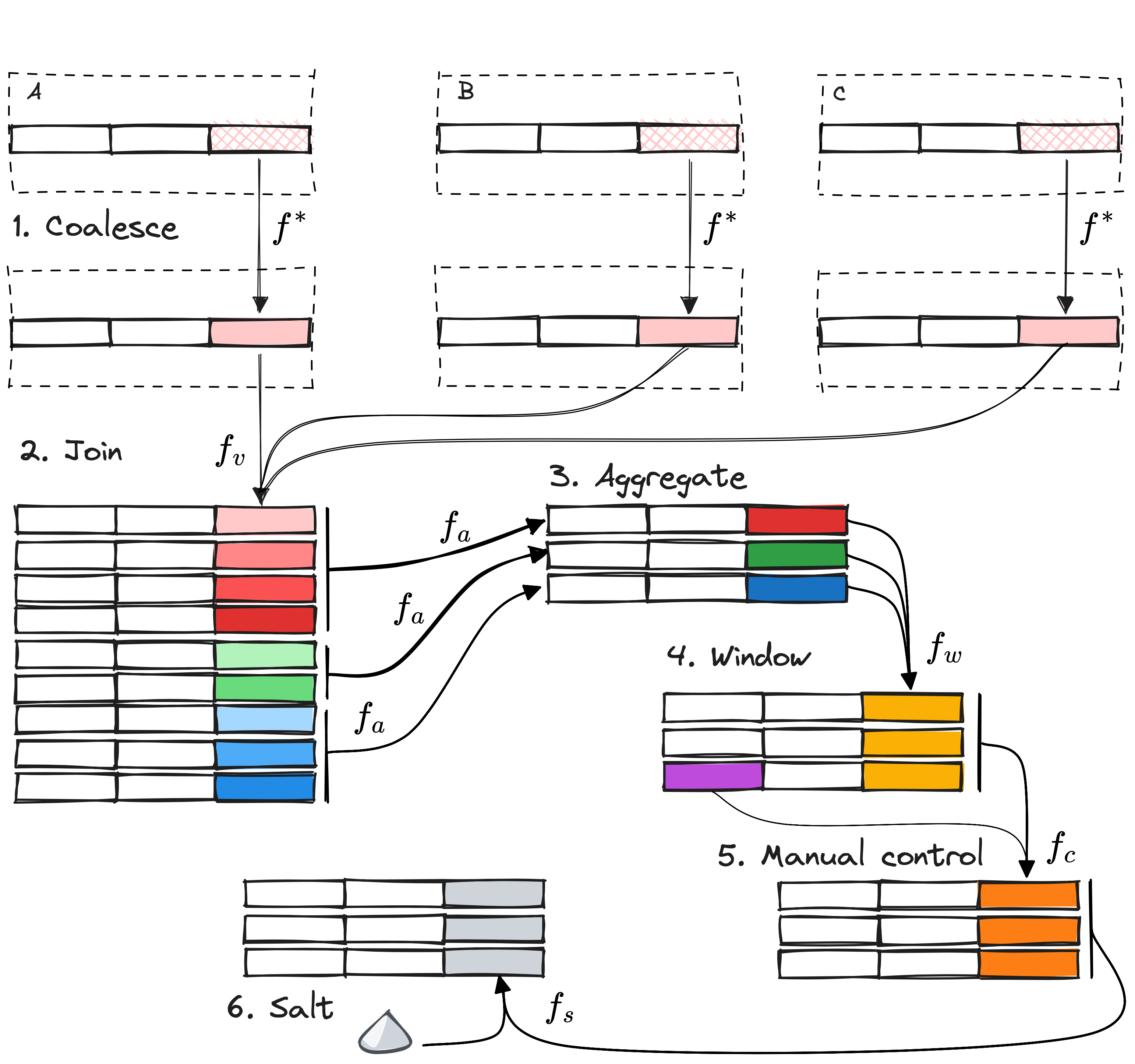}
    \caption{Overview of the token calculation for a SELECT query joining three tables $A$, $B$, and $C$. It starts with the values present in the \texttt{hash} columns of the tables involved in the query. These values are coalesced (1), combined horizontally (2), aggregated by groups (3), and aggregated vertically into a single value (4); the latter may be combined with a value extracted from the other columns of the result table (5), before being finally salted (6).}
    \Description{The diagram depicts a multi-step process starting from hashed values in three tables, followed by a series of transformations including coalescing, joining, aggregating, applying window functions, optional manual control, and finally salting the values to produce the final result.}
    \label{fig:complete_formula}
\end{figure}

Fig.~\ref{fig:complete_formula} outlines the complete process of fingerprinting a data retrieval query. It consists of a pipeline of six steps, each involving a dedicated function with its own properties:

\begin{enumerate}
\item \textbf{Coalescing function $f^*$.} As the original hash values are calculated from the rows of the tables, they cannot be NULL. However, when the main FROM clause involves an \emph{outer} join, NULL values may appear in the resulting \texttt{hash} columns. Coalescing these values with an arbitrary constant is mandatory to prevent their propagation to the final token.

\item \textbf{Variadic function $f_v$.} If $n>1$ tables are joined in the main FROM clause, the hash values of the rows of each table should be combined into a single value: $f_v(T_1\texttt{.hash}, T_2\texttt{.hash}, ..., T_n\texttt{.hash})$. Since the students cannot guess in which order the tables are enumerated in the solution, this function must be symmetric (i.e., order-insensitive).

\item \textbf{Aggregation function $f_a$.} When the outer query contains a GROUP BY clause, each group requires its own hash value. This function also must be symmetric, as the order of aggregation operations within each group is not deterministic.

\item \textbf{Window function $f_w$.} As soon as the result has multiple rows, the hashes need to be combined and distributed across the entire column. Here again, the function must be symmetric, as the order of the rows is not deterministic.

\begin{figure}[!ht]
    \includegraphics{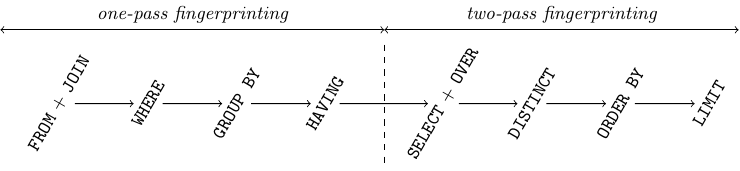}
    \caption{Logical order of execution of the clauses of an SQL query\protect\footnotemark. As the token is calculated in the SELECT clause, a single pass is sufficient to fingerprint the clauses executed \emph{before} this clause. For the other ones, the formula must include a placeholder and an instruction to manually replace it by a control value.}
    \Description{The figure divides the clauses of an SQL query into two groups: one-pass fingerprinting (FROM + JOIN, WHERE, GROUP BY, HAVING) and two-pass fingerprinting (SELECT + OVER, DISTINCT, ORDER BY, LIMIT). A dashed vertical line separates the two groups, indicating the point where the SELECT clause and subsequent clauses are processed.}
    \label{fig:clause_order}
\end{figure}
\footnotetext{
        The actual physical order can vary across DBMSs for performance or convenience reasons:
        for example, even though WHERE theoretically comes after FROM, filtering might be applied as soon as possible to reduce the dataset size; or, in the HAVING clause, references to aliases defined in the SELECT list may be tolerated (this is the default behavior of MySQL \href{https://dev.mysql.com/doc/refman/8.4/en/group-by-handling.html}{\cite{mysql}\texttt{/group-by-handling.html}}).
        However, as far as we know, none of these implementation details has any impact on the demarcation between one- and two-pass fingerprinting.
}

\item \textbf{Control function $f_c$.} When a calculation should be executed in the outer SELECT clause or after it (Fig.~\ref{fig:clause_order}), the formula includes a placeholder which the students are instructed to replace by a control value (e.g., “the third value in the \texttt{salary} column”, “the number of rows”, “the format string”, etc.). The function $f_c$ then combines this value with the one distributed in the previous stage by the window function $f_w$.

\item \textbf{Salt functions $f_s$.} If the token was constructed from a query's source-code and result table only, a wrong query could accidentally unlock a feedback associated with a different task. To avoid this, each task comes with its own unique salt function $f_s$. This ensures that the same query will produce different tokens for different tasks.

\end{enumerate}

\label{sec:global-properties}
\paragraph*{Properties of the composition.} Together, these six functions will be used to fingerprint a wide range of SQL queries, from simple data retrieval to complex aggregation and grouping operations. Their composition itself must therefore exhibit the following basic properties of a hash function:

\begin{description}
\item[Fixed-size output.] Not only must the final token belong to a certain integer range, but no step in the calculation must cause an overflow error.
\item[Determinism.] On the same database, the same query has always the same fingerprint. Once again, this implies that the functions $f_v$, $f_a$, and $f_w$ are symmetric.
\item[Few collisions.] Two queries producing different (unordered) result tables should have different fingerprints. We are not especially interested in the uniformity of the distribution of the tokens, as long as the probability of collision is low.
\item[Efficiency.] The calculation of the fingerprint must not significantly slow down that of the query.
\end{description}

On the other hand, our security requirements are modest: preventing the students from decrypting a feedback message by accident or through trial and error. We assume that their goal is to learn SQL, not to turn their database practical into a hacking challenge. First and foremost, the fingerprinting mechanism is intended to be a reliable way to check theirs answers against the predicted ones. We don't expect it to satisfy some properties commonly sought in cryptographic hash functions, such as avalanche effect, pre-image resistance, second pre-image resistance, collision resistance, pseudo-randomness, etc.

\subsubsection{Token formulas}
The composition pipeline of Fig.~\ref{fig:complete_formula} is mathematically described by Formula~\ref{eqn:formulas:agg-post}. It involves all the functions introduced above. However, in many cases, some of them can or must be dropped, leading to five formulas:

\begin{subequations}
\begin{align}
f_s\bigg(
    \phantom{f_c\Big(x, }
        f_w\big(
            \phantom{f_a\big(}
                f_v\big(
                    \{ f^*(h) \mid h \in \text{hashes} \}
                \big)
            \phantom{\big)}
        \big)
    \phantom{\Big)}
\bigg)
\tag{\theparentequation.1}
\label{eqn:formulas:basic} \\
f_s\bigg(
    \phantom{f_c\Big(x, }
          f_w\big(
            f_a\big(
                  f_v\big(
                      \{ f^*(h) \mid h \in \text{hashes} \}
                  \big)
            \big)
          \big)
    \phantom{\Big)}
\bigg)
\tag{\theparentequation.2}
\label{eqn:formulas:agg} \\
f_s\bigg(
    f_c\Big(x, 
        f_w\big(
            \phantom{f_a\big(}
                f_v\big(
                    \{ f^*(h) \mid h \in \text{hashes} \}
                \big)
              \phantom{\big)}
        \big)
    \Big)
\bigg)
\tag{\theparentequation.1c}
\label{eqn:formulas:basic-post} \\
f_s\bigg(
    f_c\Big(x, 
        f_w\big(
            f_a\big(
                f_v\big(
                    \{ f^*(h) \mid h \in \text{hashes} \}
                \big)
              \big)
        \big)
    \Big)
\bigg)
\tag{\theparentequation.2c}
\label{eqn:formulas:agg-post}\\
f_s\bigg(
    f_c\Big(x, 
        \makebox[0pt][l]{0}
        \phantom{f_w\big(
            f_a\big(
                f_v\big(
                    \{ f^*(h) \mid h \in \text{hashes} \}
                \big)
              \big)
        \big)}
    \Big)
\bigg)
\tag{\theparentequation.c}
\label{eqn:formulas:expression-only}
\end{align}
\label{eqn:formulas:all}
\end{subequations}

In our numbering scheme, a final \emph{c} denotes the inclusion of a control function $f_c$. Formulas~\ref{eqn:formulas:basic-post} and \ref{eqn:formulas:agg-post} are “controlled” variants of Formulas~\ref{eqn:formulas:basic} and \ref{eqn:formulas:agg}, respectively. The \textbf{dimension} (of an instance) \textbf{of a formula} is the arity of its variadic function $f_v$. For example, the instance of Formula~\ref{eqn:formulas:agg-post} presented in Fig.~\ref{fig:complete_formula} is 3-dimensional (it involves three tables $A$, $B$, and $C$). All instances of Formula~\ref{eqn:formulas:expression-only} are 0-dimensional.

\subsection{Fingerprinting scope}

\subsubsection{Main scope} \label{sec:main-scope}

Let us now examine which types of statements are covered by which fingerprinting mechanism.

\noindent\begin{minipage}{\textwidth}
\setlength\intextsep{-\baselineskip}
\begin{wrapfigure}{R}{0pt}
\begin{lstlisting}
    SELECT scalar_expressions,
           formula_§\ref{eqn:formulas:basic}§
    FROM   joined_tables
    WHERE  conditions
    §~§
    §~§
\end{lstlisting}
\end{wrapfigure}
\paragraph*{\hspace{1em}Basic queries.} Formula \ref{eqn:formulas:basic} is designed for the first type of query usually presented in an SQL course: it consists of a SELECT clause followed by a FROM clause, and an optional WHERE clause. Here, “basic” qualifies the outer query. In fact, nothing prevents it to contain subqueries (in the SELECT or the WHERE clauses) or derived queries (in the FROM clause). However, one must keep in mind that these nested parts are a “black box” for the outer SELECT's formula.
\end{minipage}\\

In any case, such queries can be fingerprinted in a single pass (no need for a control function $f_c$). As they don't involve any grouping or aggregation, the aggregation function $f_a$ \emph{cannot} be applied.

\paragraph*{Basic queries with post-processing.}
Formula \ref{eqn:formulas:basic-post} is intended for those queries that involve deduplication, ordering, and/or limiting. Such operations are applied after the calculation of the SELECT's expressions (cf. Fig.~\ref{fig:clause_order}), including the token formula. Therefore, they need the control function $f_c$ to take into account some chosen defining characteristics of the result table.

\begin{tabular}{p{0.35\textwidth}p{0.3\textwidth}p{0.25\textwidth}}
\begin{lstlisting}
SELECT DISTINCT scalar_exprs,
                formula_§\ref{eqn:formulas:basic-post}§
FROM   joined_tables
WHERE  conditions
§~§
\end{lstlisting}
&
\begin{lstlisting}
SELECT   scalar_exprs,
         formula_§\ref{eqn:formulas:basic-post}§
FROM     joined_tables
WHERE    conditions
ORDER BY columns
\end{lstlisting}    
&
\begin{lstlisting}
SELECT scalar_exprs,
       formula_§\ref{eqn:formulas:basic-post}§
FROM   joined_tables
WHERE  conditions
LIMIT  row_count
\end{lstlisting}
\end{tabular}

\noindent\begin{minipage}{\textwidth}
\setlength\intextsep{-\baselineskip}
\begin{wrapfigure}{R}{0pt}
\begin{lstlisting}
    SELECT aggregate_exprs,
           formula_§\ref{eqn:formulas:agg-post}§
    FROM   joined_tables
    WHERE  conditions
    §~§
    §~§
\end{lstlisting}
\end{wrapfigure}
\paragraph*{\hspace{1em}Queries with scalar aggregates.}
A \textbf{scalar aggregation} \cite{bowman2001,galindo2001} is a query that returns one single row consisting of one or several aggregated values. Such simple queries might be used to assess students' understanding of a given aggregation function. Full-form Formula~\ref{eqn:formulas:agg-post} is used here: the outcome of the aggregation has to be checked by a control function $f_c$, as it is calculated in the SELECT clause, making it unavailable in the first pass of the token calculation (cf. Fig.~\ref{fig:clause_order}).
\end{minipage}\\

Note that, since the aggregation outcome is scalar, the window function $f_w$ is unary idempotent \cite{marichal2015}, i.e., it has no effect on the final token and could be omitted. However, we recommend including it for consistency. Using a specific variation of Formula~\ref{eqn:formulas:agg-post} could, after a certain number of similar tasks, “leak” to the most astute students its structural correlation with the type of the expected query.

\noindent\begin{minipage}{\textwidth}
\setlength\intextsep{-\baselineskip}
\begin{wrapfigure}{R}{0pt}
\begin{lstlisting}
    SELECT   scalar_exprs, aggregate_exprs,
             formula_§\ref{eqn:formulas:agg}/\ref{eqn:formulas:agg-post}§
    FROM     joined_tables
    WHERE    conditions
    GROUP BY columns
    HAVING   conditions
    §~§
    §~§
\end{lstlisting}
\end{wrapfigure}
\paragraph*{\hspace{1em}Queries with vector aggregates.}
A \textbf{vector aggregation} \cite{bowman2001,galindo2001} is a query which specifies grouping columns as well as aggregates to be calculated for each group. Two formulas are possible here: \ref{eqn:formulas:agg} and \ref{eqn:formulas:agg-post}. Choosing one depends on the pedagogical purpose of the task. If the focus is on the grouping skills, the one-pass formula \ref{eqn:formulas:agg} is sufficient: the control function $f_c$ requiring an additional manual step, it is better to avoid it when possible.
If the focus is on the aggregation skills (or both), we fall back to Formula~\ref{eqn:formulas:agg-post}.
\end{minipage}\\

\paragraph*{Queries with or without top level aggregation.}
In \cite{date1997}, Date mentions that “the GROUP BY and HAVING clauses are actually redundant\,---\,i.e., for every select-expression that involves such a clause, there is a semantically equivalent one that does not”.

\noindent\begin{minipage}{\textwidth}
\setlength\intextsep{-\baselineskip}
\begin{wrapfigure}{R}{0pt}
\begin{lstlisting}
    SELECT DISTINCT dpt_id,
                    (SELECT avg(salary)
                    FROM employee B
                    WHERE A.dpt_id = B.dpt_id)
    FROM employee A
    §~§
    §~§
\end{lstlisting}
\end{wrapfigure}
For instance, \SELECT \verb|dpt_id, avg(salary)| \FROM \verb|employee| \GROUPBY \verb|dpt_id| could be rewritten with a correlated subquery as in the right-hand side query.
Although the queries are semantically equivalent, they cannot be fingerprinted in the same way. One calls for Formula~\ref{eqn:formulas:agg-post}, while the second one requires Formula~\ref{eqn:formulas:basic-post}. The two main categories of formulas (\ref{eqn:formulas:basic} and \ref{eqn:formulas:agg}) are never interchangeable. This is a case where choosing one over the other is not dictated by the semantics of the query, but by the instructor's intent to assess the mastery of a certain style. \hfill\accmat 
\end{minipage}\\

\noindent\begin{minipage}{\textwidth}
\setlength\intextsep{-\baselineskip}
\begin{wrapfigure}{R}{0pt}
\begin{lstlisting}
    SELECT scalar_exprs,
           formula_§\ref{eqn:formulas:expression-only}§
    [FROM  DUAL]
    §~§
    §~§
\end{lstlisting}
\end{wrapfigure}
\paragraph*{\hspace{1em}Expression-only queries.} 
One might want to test students' ability to look up and understand the documentation for complex scalar functions, such as \verb!date_format()! in MySQL or \verb!to_char()! in PostgreSQL (they are often vendor-specific). To this purpose, a query reduced to a scalar expression, with no FROM clause (in most RDBMSs) or FROM DUAL (in Oracle), is usually sufficient. In this situation, the most degenerate formula (\ref{eqn:formulas:expression-only}) applies.
\end{minipage}\\

\subsubsection{Extended scope} \label{sec:extended-scope}

Some difficulties and their workarounds are presented in the following paragraphs.

\paragraph*{DML statements.}
The simplest way to fingerprint an UPDATE, DELETE, or INSERT statement is to check the result table after the modification has been applied. This post-processing step can be done by executing \SELECT Formula~\ref{eqn:formulas:basic} \FROM the affected table.

\begin{tabular}{p{0.3\textwidth}p{0.3\textwidth}p{0.3\textwidth}}
\begin{lstlisting}
UPDATE table_name
SET    columns
WHERE  conditions
;
SELECT formula_§\ref{eqn:formulas:basic}§
FROM   table_name
\end{lstlisting}
&
\begin{lstlisting}
DELETE FROM table_name
WHERE  conditions
;

SELECT formula_§\ref{eqn:formulas:basic}§
FROM   table_name
\end{lstlisting}    
&
\begin{lstlisting}
INSERT INTO table_name
VALUES (new_values_1),
       (new_values_2)... 
;
SELECT formula_§\ref{eqn:formulas:basic}§
FROM   table_name
\end{lstlisting}
\end{tabular}

As far as we know, writing a DML statement that would directly produce the resulting token is not possible\footnote{
    As of June 2024, some DBMSs (PostgreSQL, SQLite, Oracle, SQL Server, etc.) support a non-standard RETURNING clause that might be thought of as a potential solution. Placed at the end of a modification command, it returns one result row for each database row that has been deleted, inserted, or updated. The keyword is followed by a list of expressions that define the values of the columns in the result table \href{https://www.sqlite.org/lang_returning.html}{\cite{sqlite}\texttt{/lang\_returning.html}}. However, window and aggregation functions are not allowed there, which limits their applicability to single-row modifications. In these cases, Formula~\ref{eqn:formulas:expression-only} may be used, but could not tell whether the checked statement were an INSERT, an UPDATE, a DELETE or even a trivial SELECT query.
}.

\noindent\begin{minipage}{\textwidth}
\setlength\intextsep{-\baselineskip}
\begin{wrapfigure}{R}{0pt}
\begin{lstlisting}
    SELECT emp_name, formula_§\ref{eqn:formulas:basic}§
    FROM   employee A
    JOIN   ( SELECT emp_id FROM employee
             EXCEPT
             SELECT manager_id FROM department
           ) B ON A.emp_id = B.emp_id
    §~§
    §~§
\end{lstlisting}
\end{wrapfigure}
\paragraph*{\hspace{1em}Derived queries with no aggregation.}
Some operations are impossible to fingerprint when they appear at the top level of a query. For example, a set operation cannot correctly work if its operands include the hash values, as they are all distinct. A workaround is to perform the set operation in a derived table, and then “restore” the hashes by joining it with the original ones. Formula~\ref{eqn:formulas:basic} can then be applied, in the same dimension and with same resulting token as an equivalent NOT EXISTS query. \hfill\accmat
\end{minipage}\\

Formula~\ref{eqn:formulas:basic} also applies when the inner SELECT clause can include a \texttt{hash} column. The raw hashes are then “transmitted” to the formula of the outer SELECT clause. This takes more than simply copying and pasting a formula, and starts to push our system to the limits of its practicality. At the very least, the problem statement should involve some degree of scaffolding to guide the students through the process \cite{wood1976}.

These two fingerprinting techniques (“hash restoration” or “raw hash transmission”) cannot be used when the derived query involves an aggregation.

\noindent\begin{minipage}{\textwidth}
\setlength\intextsep{-\baselineskip}
\begin{wrapfigure}{R}{0pt}
\begin{lstlisting}
    SELECT aggregate_exprs,
           formula_§\ref{eqn:formulas:expression-only}§
    FROM   ( SELECT   scalar_exprs,
                      aggregate_exprs
             FROM     joined_tables
             GROUP BY columns
           ) A
    WHERE  conditions
    §~§
    §~§
\end{lstlisting}
\end{wrapfigure}
\paragraph*{\hspace{1em}Derived queries with aggregation.}
SQL has no direct support for multi-level aggregations: \SELECT \texttt{max(avg(salary))} is incorrect; average salaries have to be calculated in a derived query, and the maximum of these averages in the outer query. With this pattern, it is impossible to “transmit” the raw hashes: they must be aggregated first, which implies breaking down the token formula between the inner and the outer SELECT clauses (see four examples in Appendix, Table~\ref{tab:six_variations}). Here, it's worth remembering that SQLab's \emph{raison d'être} is to get students learning and practicing SQL, \emph{not} mastering the intricacies of our fingerprinting mechanism. We thus recommend a simple workaround: drop the hash values altogether, and only rely on a control value copied from the result table and pasted into Formula~\ref{eqn:formulas:expression-only}.
\end{minipage}\\

\paragraph*{Common Table Expressions.}
The queries studied in the previous two paragraphs can be rewritten using one or more CTE. The considerations relative to the fingerprint still hold. The hash transmission technique works as well with \emph{recursive} CTE. \hfill\accmat

\subsubsection{Out-of-scope area} \label{sec:out-of-scope}

The latter cases were arguably lying at the edge of the practicality, but the following categories of SQL statements are definitely off-limits, i.e., cannot be fingerprinted by our mechanism:
\begin{itemize}
\item queries, both correct and incorrect, that return an empty table;
\item queries that raise a syntax error or a runtime error without returning a result table;
\item statements that pertain to other SQL sublanguages, such as DDL (CREATE, ALTER, DROP), DCL (GRANT, REVOKE), or TCL (COMMIT, ROLLBACK, SAVEPOINT).
\end{itemize}

\subsection{Fingerprinting accuracy} \label{sec:accuracy}
In this section, we assume that a fingerprint can be calculated, and proceed to evaluate the capacity of our mechanism to assess the correctness (or lack thereof) of an answer. 

\subsubsection{False negatives} Avoiding situations where a correct query is fingerprinted as incorrect entirely falls to the designer of the task: as long as she has marked as correct the token(s) produced by the queries she considers correct, these same queries will deterministically produce the same tokens, and thus be found correct by the system. Does she need to think of all the possible correct queries her students might write? Not necessarily: two gold queries that produce the same token are a rare exception rather than the rule.

\subsubsection{False positives} \label{sec:false-positives}
They occur when an incorrect query is fingerprinted as correct. As we know, the vast majority of incorrect queries result in an incorrect table, i.e., a different table than the gold one. The good news is that a token can more or less be regarded as a hash of the result table: if the tables are different, the tokens will be different. As we will see in Theorem~\ref{thm:accuracy}, this assertion needs to be nuanced, but it has a compelling corollary: an incorrect query is easier to detect when it actually produces an incorrect table. This generally depends on the database state, and thus on the designer's commitment to introduce enough “chaff” to make incorrect queries yield a result distinct, not only from the gold query, but also from each other. The latter property is key to the relevance of the hint which will be triggered by a wrong answer. However, there still remains another types of false positives, which are entirely intrinsic to the fingerprinting mechanism itself.

\begin{table}
    \tablecaption{
        All of these queries address the same question: “Which employee works 5 hours on project 30?”. They all return the same single-cell table, \texttt{'Ahmad V. Jabbar'}. For formatting reasons, USING is preferred to ON, and the condition “\texttt{hours = 5} \AND \texttt{prj\_id = '30'}” is replaced with “\texttt{etc.}” if needed.
        Before execution, the SELECT clause is augmented with either the one-dimensional or two-dimensional form of the basic Formula~\ref{eqn:formulas:basic}. The resulting tokens are displayed in columns $t_1$ and $t_2$, respectively. For better readability, actual token values are bijectively replaced with a counter. Only odd numbers are used for $t_1$, and even numbers for $t_2$.
        If a query lacks one of the required tables \texttt{A} and/or \texttt{B}, the token calculation raises an error, represented by a blank cell.
        Queries are identified by a letter and a number. Letter $G$ denotes a $G$old answer. Here, such a query can produce three different tokens: 1, 2, and 3.
        Letter $O$ denotes an $O$verly complicated correct answer. If such a query produces a token 1, 2, or 3, it indicates that the complication was not detected by the fingerprinting mechanism: we highlight this token in \textcolor{ACMOrange}{orange}.
        Letter $I$ denotes an $I$ncorrect answer, i.e., a query that might result in a wrong table for a different state of the database. If such a query produces a token 1, 2, or 3, it indicates that the fingerprinting mechanism failed to detect the error on the \emph{actual} state: we highlight this token in \textcolor{ACMRed}{red}. \hfill\accmat
    }
    
    \small\sffamily
    \begin{tabular}{l|cc|l}
    id. & $t_1$ & $t_2$ & raw query (before the injection of Formula~\ref{eqn:formulas:basic}) \\
    \hline
    $G_{1}$ & 1 &  & \lstinline[basicstyle=\scriptsize\ttfamily,keepspaces=true]|SELECT emp_name FROM employee A WHERE emp_id IN (SELECT emp_id FROM works_on WHERE etc.)| \\
\cline{1-3}
$G_{2}$ & 1 & 2 & \lstinline[basicstyle=\scriptsize\ttfamily,keepspaces=true]|SELECT A.emp_name FROM employee A JOIN works_on B USING (emp_id) WHERE hours = 5 AND prj_id = '30'| \\
$G_{3}$ & 3 & 2 & \lstinline[basicstyle=\scriptsize\ttfamily,keepspaces=true]|SELECT B.emp_name FROM employee B JOIN works_on A USING (emp_id) WHERE hours = 5 AND prj_id = '30'| \\
\cline{1-3}
$G_{4}$ & 1 & 2 & \lstinline[basicstyle=\scriptsize\ttfamily,keepspaces=true]|SELECT A.emp_name FROM employee A, works_on B WHERE A.emp_id = B.emp_id AND etc.| \\
$G_{5}$ & 3 & 2 & \lstinline[basicstyle=\scriptsize\ttfamily,keepspaces=true]|SELECT B.emp_name FROM employee B, works_on A WHERE A.emp_id = B.emp_id AND etc.| \\
\hline
$O_{1}$ & \textcolor{ACMOrange}{1} & \textcolor{ACMOrange}{2} & \lstinline[basicstyle=\scriptsize\ttfamily,keepspaces=true]|SELECT A.emp_name FROM employee A RIGHT JOIN works_on B USING (emp_id) WHERE etc.| \\
$O_{2}$ & \textcolor{ACMOrange}{3} & \textcolor{ACMOrange}{2} & \lstinline[basicstyle=\scriptsize\ttfamily,keepspaces=true]|SELECT B.emp_name FROM employee B RIGHT JOIN works_on A USING (emp_id) WHERE etc.| \\
\cline{1-3}
$O_{3}$ & \textcolor{ACMOrange}{1} & 4 & \lstinline[basicstyle=\scriptsize\ttfamily,keepspaces=true]|SELECT A.emp_name FROM employee A JOIN employee B USING (emp_id) JOIN works_on O USING (emp_id) WHERE etc.| \\
$O_{4}$ & \textcolor{ACMOrange}{3} & \textcolor{ACMOrange}{2} & \lstinline[basicstyle=\scriptsize\ttfamily,keepspaces=true]|SELECT B.emp_name FROM employee B JOIN employee O USING (emp_id) JOIN works_on A USING (emp_id) WHERE etc.| \\
$O_{5}$ & \textcolor{ACMOrange}{1} & \textcolor{ACMOrange}{2} & \lstinline[basicstyle=\scriptsize\ttfamily,keepspaces=true]|SELECT O.emp_name FROM employee O JOIN employee A USING (emp_id) JOIN works_on B USING (emp_id) WHERE etc.| \\
\cline{1-3}
$O_{6}$ & \textcolor{ACMOrange}{1} & \textcolor{ACMOrange}{2} & \lstinline[basicstyle=\scriptsize\ttfamily,keepspaces=true]|SELECT A.emp_name FROM employee A JOIN works_on B USING (emp_id) JOIN project O USING (prj_id) WHERE etc.| \\
$O_{7}$ & 5 & 6 & \lstinline[basicstyle=\scriptsize\ttfamily,keepspaces=true]|SELECT B.emp_name FROM employee B JOIN works_on O USING (emp_id) JOIN project A USING (prj_id) WHERE etc.| \\
$O_{8}$ & \textcolor{ACMOrange}{3} & 8 & \lstinline[basicstyle=\scriptsize\ttfamily,keepspaces=true]|SELECT O.emp_name FROM employee O JOIN works_on A USING (emp_id) JOIN project B USING (prj_id) WHERE etc.| \\
\hline
$I_{1}$ & \textcolor{ACMRed}{1} & \textcolor{ACMRed}{2} & \lstinline[basicstyle=\scriptsize\ttfamily,keepspaces=true]|SELECT A.emp_name FROM employee A LEFT JOIN works_on B USING (emp_id) WHERE etc.| \\
$I_{2}$ & \textcolor{ACMRed}{3} & \textcolor{ACMRed}{2} & \lstinline[basicstyle=\scriptsize\ttfamily,keepspaces=true]|SELECT B.emp_name FROM employee B LEFT JOIN works_on A USING (emp_id) WHERE etc.| \\
\cline{1-3}
$I_{3}$ &  &  & \lstinline[basicstyle=\scriptsize\ttfamily,keepspaces=true]|SELECT 'Ahmad V. Jabbar' | \\
\cline{1-3}
$I_{4}$ & \textcolor{ACMRed}{1} &  & \lstinline[basicstyle=\scriptsize\ttfamily,keepspaces=true]|SELECT emp_name FROM employee A WHERE emp_name = 'Ahmad V. Jabbar'| \\
$I_{5}$ & \textcolor{ACMRed}{1} &  & \lstinline[basicstyle=\scriptsize\ttfamily,keepspaces=true]|SELECT emp_name FROM employee A WHERE emp_id = '987987987'| \\
\cline{1-3}
$I_{6}$ & 7 &  & \lstinline[basicstyle=\scriptsize\ttfamily,keepspaces=true]|SELECT DISTINCT 'Ahmad V. Jabbar' FROM employee A| \\
$I_{7}$ & 9 & 10 & \lstinline[basicstyle=\scriptsize\ttfamily,keepspaces=true]|SELECT DISTINCT 'Ahmad V. Jabbar' FROM employee A, works_on B| \\
\cline{1-3}
$I_{8}$ & \textcolor{ACMRed}{1} & \textcolor{ACMRed}{2} & \lstinline[basicstyle=\scriptsize\ttfamily,keepspaces=true]|SELECT emp_name FROM employee A JOIN works_on B USING (emp_id) WHERE hours = 5| \\
\cline{1-3}
$I_{9}$ & \textcolor{ACMRed}{1} & 6 & \lstinline[basicstyle=\scriptsize\ttfamily,keepspaces=true]|SELECT A.emp_name FROM employee A JOIN project B USING (dpt_id) WHERE prj_id = '30' and sex = 'M'| \\
$I_{10}$ & 5 & 6 & \lstinline[basicstyle=\scriptsize\ttfamily,keepspaces=true]|SELECT B.emp_name FROM employee B JOIN project A USING (dpt_id) WHERE prj_id = '30' and sex = 'M'| \\
    \end{tabular}
    
    \label{tab:same_result_queries_basic}
\end{table}

\subsubsection{Worst-case example for false positives} \label{sec:worst-case}
Let us create an example of a task with several gold and wrong answers, all resulting in the same table. To maximize the number of possible “collisions”, the result is as small as possible: a single-cell table. Note that execution matching would be useless here, as all queries would be marked as correct.

Table~\ref{tab:same_result_queries_basic} presents possible queries to identify the unique employee who works 5 hours on project 30. This task requires filtering the \texttt{works\_on} table and to join it with the \texttt{employee} table. This falls into the scope of Formula~\ref{eqn:formulas:basic}. In the one-dimensional case (column $t_1$), students can express the join either in the WHERE clause ($G_1$), or in the FROM clause ($G_2$, ..., $G_5$). In the two-dimensional case, $G_1$ raises an error (as the \texttt{works\_on} table is unknown at the top level), while symmetric queries $G_2$, $G_3$ and $G_4$, $G_5$ produce the same token (as the variadic function $f_v$ itself is symmetric).

Constructing overly complex yet correct queries is easy. $O_1$ and $O_2$ are variations of $G_2$ and $G_3$, respectively, with a RIGHT keyword in the join that is both unnecessary and undetectable, since \texttt{emp\_id} is a not-nullable foreign key in \texttt{works\_on}. Introducing a useless self-join ($O_3$, $O_4$, $O_5$) or an extraneous third table ($O_6$, $O_7$, $O_8$) might be detected or not, depending on whether the formula refers to the offending table.

In the database state of Fig.~\ref{fig:company}, many incorrect queries can yield the expected single-cell result. 
Queries $I_1$ and $I_2$ are symmetric of $O_1$ and $O_2$, respectively. They would produce a different table if at least one employee were not working on any project.
Queries $I_3$ to $I_7$ exhibit \emph{condition-based value constraints}, a subcategory of deceptive errors that “denotes the infusion of magic numbers [...] with the intent to selectively sieve records that align with the OJS’s projected outcomes” \cite{wang2024}. Although query $I_3$ raises an error when complemented with formula, queries $I_4$ and $I_5$ are indistinguishable from query $G_1$. However, desperate uses of DISTINCT ($I_6$ and $I_7$) are detected because DISTINCT is evaluated \emph{after} the formula.
Erroneous conditions ($I_8$, $I_9$, $I_{10}$) may or may not be detected, depending on the state of the database. As we have said, it is the designer's responsibility to introduce “chaff” until every reasonably incomplete or incorrect condition yields a unique result.

It is worth noting that, for a mutation analysis based on the comparison of result tables, every overly complex correct queries $O_\#$ is an equivalent (i.e., unkillable) mutant. Nevertheless, our fingerprinting mechanism is able to detect either 1/8th or 3/8th of them, depending on the choice of the formula. Concerning the incorrect queries $I_\#$, although none is killed by the database state, 4/10th or 7/10th of them are already detected by our mechanism.

\subsubsection{Accuracy theorem} \label{sec:theory}

To understand how queries that return the same result table can produce different tokens, we need to introduce a simple query transformation.

\begin{definition}[Starring]
\begin{enumerate}
\item Let $Q$ be a query compatible with a $n$-dimensional instance $F_n$ of Formula~\ref{eqn:formulas:basic}, and let $T_1$, $T_2$, ..., $T_n$ be the tables involved by $F_n$. We call \textbf{starring}, the transformation that suppresses the clauses and operations executed after the outer SELECT clause (cf. Fig.~\ref{fig:clause_order}), and replaces the outer SELECT clause with $T_1\mathtt{.*}, T_2\mathtt{.*}, ..., T_n\mathtt{.*}$ (where $\mathtt{.*}$ is SQL notation). The result is the \textbf{starred query} $Q^\mathtt{*}$.
\item For Formula~\ref{eqn:formulas:agg}-compatible queries, the replacement expression is $\mathcal{C}(T_1\mathtt{.*}, T_2\mathtt{.*}, ..., T_n\mathtt{.*})$, where $\mathcal{C}$ is a vectorized null-preserving concatenation function.
\end{enumerate}
\end{definition}

Example for a two-dimensional Formula~\ref{eqn:formulas:basic}. $Q$ and $Q^\mathtt{*}$ return 16 and 13 rows, respectively: \hfill\accmat

\begin{tabular}{ll}
$Q$
&
\begin{lstlisting}
SELECT DISTINCT emp_id, location
FROM works_on A JOIN project B ON A.prj_id = B.prj_id
ORDER BY emp_id
\end{lstlisting}
\\
$Q^\mathtt{*}$
&
\begin{lstlisting}
SELECT A.*, B.*
FROM works_on A JOIN project B ON A.prj_id = B.prj_id
\end{lstlisting}  
\end{tabular}

Example for a one-dimensional Formula~\ref{eqn:formulas:agg}. $Q$ and $Q^\mathtt{*}$ both return 4 rows: \hfill\accmat

\begin{tabular}{ll}
$Q$
&
\begin{lstlisting}
SELECT location, count(*)
FROM project A GROUP BY location
\end{lstlisting}
\\
$Q^\mathtt{*}$
&
\begin{lstlisting}
SELECT §$f$§(A.prj_name), §$f$§(A.prj_id), §$f$§(A.location), §$f$§(A.dpt_id), §$f$§(A.hash)
FROM project A GROUP BY location
\end{lstlisting}
\end{tabular}

... where $f$ is, for instance, \verb!json_arrayagg! (MySQL, Oracle, SQL Server, etc.).

\begin{definition}[Perfect setting] The combination of a database state $S$ and the functions involved in the token calculation is said to be \textbf{perfect} if the following conditions are met:
\begin{enumerate}
\item Hash function $h$ is injective on $S$, i.e., all values on the columns \verb!hash! are distinct positive integers.
\item Coalescing function $f^*$ maps the NULL value to a positive integer that is distinct from actual outputs of $h$.
\item Variadic function $f_v$ is injective on the set of all combinations with repetitions of actual outputs of $h$ and $f^*$.
\item Aggregation function $f_a$, window function $f_w$, control function $f_c$, and salt function $f_s$ are injective on the set of actual outputs of $f_v$, $f_a$, $f_w$, and $f_c$, respectively.
\end{enumerate}
\end{definition}

Only the first two conditions are easy to check during the database construction. Note however that the goal of these conditions is to guarantee that the fingerprinting mechanism has no collision across all existing queries associated with distinct feedbacks. For this purpose, it is sufficient to check \emph{a posteriori} that the generated tokens are indeed distinct. While the constraints of a perfect setting are too strong to be useful in practice, they constitute a necessary theoretical framework to establish the main result of this paper:

\begin{theorem}[Fingerprinting accuracy]\label{thm:accuracy}
In a perfect setting, two queries $Q_1$ and $Q_2$ with the same formula \ref{eqn:formulas:basic} or \ref{eqn:formulas:agg} and no derived table in their FROM clause produce the same token if and only if their starred versions $Q_1^\mathtt{*}$ and $Q_2^\mathtt{*}$ result in the same table (ignoring the null columns and the order of columns, rows, and in-group rows).
\end{theorem}

\begin{proof}
    See Annex~\ref{sec:proof-accuracy}.
\end{proof}

Consider two examples of queries that produce the same table, but different tokens:

\begin{enumerate}

\item \begin{tabular}{p{0.37\textwidth}p{0.52\textwidth}}
\begin{lstlisting}
SELECT count(*) as subordinates
     , formula_§\ref{eqn:formulas:agg-post}§
FROM employee
WHERE supervisor_id IS NOT NULL
\end{lstlisting}    
&
\begin{lstlisting}
SELECT count(supervisor_id) as subordinates
     , formula_§\ref{eqn:formulas:agg-post}§
FROM employee
\end{lstlisting}
\end{tabular}~\\[-\bigskipamount]

Their starred queries indeed produce different tables: the left one aggregates 7 rows (those with a non-NULL \texttt{supervisor\_id}), while the right one aggregates all 8 rows. Since the original queries are semantically equivalent, the instructor should mark one as the gold query, and the other as a variant. \hfill\accmat

\item \begin{tabular}{p{0.47\textwidth}p{0.41\textwidth}}
\begin{lstlisting}
SELECT emp_id, salary, formula_§\ref{eqn:formulas:basic-post}§
FROM employee
WHERE salary = (SELECT max(salary)
                FROM employee)
\end{lstlisting}    
&
\begin{lstlisting}
SELECT emp_id, salary, formula_§\ref{eqn:formulas:basic-post}§
FROM employee
ORDER BY salary DESC
LIMIT 1
\end{lstlisting}
\end{tabular}~\\[-\bigskipamount]

The starred form of the left query returns a single-row table, while the right one loses its LIMIT clause and returns 8 rows. Although the two original queries return the same table \emph{in this state of the database}, they are not equivalent. The difference in the tokens enables the instructor to mark the left query as gold, and the right query as incorrect. \hfill\accmat

\end{enumerate}

\subsubsection{Extensions and limitations}

The preceding analysis should have clarified how our model diverges from the conventional OJS paradigm: it performs execution matching on the result tables of the \emph{starred queries}, as opposed to the original queries. These starred tables contain \emph{more} information, as they “record” which rows of which source tables are involved in the execution of the (outer) query. Nevertheless, they also may contain \emph{less} information, namely, they “ignore” every operation executed during or after the outer SELECT clause: projected aggregations, deduplication, row ordering, row limiting, etc. (right part of Fig.~\ref{fig:clause_order}).

\paragraph*{Extensions.}
The controlled versions \ref{eqn:formulas:basic-post} and  \ref{eqn:formulas:agg-post} of the two main formulas were introduced to address these shortcomings, albeit at the expense of a second pass, that requires the student to copy and paste a designated value. They strictly extend the scope of Theorem~\ref{thm:accuracy} to encompass queries where these operations are deemed significant. However, the accuracy of the fingerprinting mechanism is still limited.

\paragraph*{Limitations relative to SQL features.}
Our model cannot assess students' mastery of the following features:
\begin{itemize}
    \item Syntactically distinct but semantically equivalent constructs, including but not limited to:

    \begin{tabular}{lclp{1em}|p{1em}lcl}
        ~\\[-\medskipamount]
        \lstinline|x BETWEEN a AND b| & $\equiv$ & \lstinline|x >= a AND x <= b| &&&
        \lstinline|foo bar| & $\equiv$ & \lstinline|foo AS bar|\\
        \lstinline|x IN ('foo', 'bar')| & $\equiv$ & \lstinline|x = 'foo' OR x = 'bar'| &&&
        \lstinline|USING (id)| & $\equiv$ & \lstinline|ON a.id = b.id| \\
        \lstinline|EXISTS (SELECT 1 ...)| & $\equiv$ & \lstinline|EXISTS (SELECT * ...)| &&&
        \lstinline|LEFT JOIN| & $\equiv$ & \lstinline|LEFT OUTER JOIN|\\
        \lstinline|WHERE x <> y| & $\equiv$ & \lstinline|WHERE x != y| &&&
        \lstinline|JOIN| & $\equiv$ & \lstinline|INNER JOIN|\\
        ~\\[-\medskipamount]
    \end{tabular}

    \item Minute details of the projection, such as the order of the resulting columns, their aliases, the absence of extraneous columns, etc.
    \item Column qualification, identifier capitalization, formatting, and other stylistic choices.
\end{itemize}

\paragraph*{Limitations relative to SQL errors.}
Our model cannot detect and/or differentiate between all types of errors. Let us describe its behavior on the query outcomes classified by Taipulus et al. in Table~\ref{tab:taipalus} \cite{taipalus2018}. Like these authors, we reuse the numbering of the “quite complete list” originally proposed in \cite{brass2006}. We provide working examples in the accompanying notebook whenever interesting or necessary. \hfill\accmat

\begin{itemize}
    \item Correct result table.
    \begin{itemize}
        \item Gold queries are covered by our mechanism within the limits of Theorem~\ref{thm:accuracy} and its two-pass extension.
        \item Complications
        2, 
        6, 
        7, 
        8, 
        11, 
        12, 
        13, 
        14, 
        17, 
        18, 
        20, 
        23, 
        24, 
        25, 
        26, 
        30, 
        35, 
        and 36 
        are undetectable.
        \item Complication 5 (unused tuple variable) is detected. 
        \item Complications
        15, 
        19, 
        and 22 
        raise a runtime error due to an incompatibility with the provided formula.
    \end{itemize}
    \item Incorrect result table.
    \begin{itemize}
        \item Logical errors can be detected as long the database state has appropriate chaff.
        \item Semantic errors
        3 
        and 4 
        are undetectable, but unimportant (constant or duplicate output columns).
        \item Semantic errors
        27, 
        28, 
        29, 
        31, 
        33, 
        34, 
        37, 
        and 38 
        are detected.
        \item Semantic errors
        1 
        and 10 
        return an empty result, and thus no token.
    \end{itemize}
    \item Syntax errors and remaining runtime errors (39-43) are not covered.
\end{itemize}

In comparison, recall that execution matching in a conventional OJS is only able to return a binary response \cite{koberlein2024}. It cannot detect any complication, nor differentiate between any two logical or semantic errors. Our model is far more nuanced. Its capability to associate a distinct fingerprint to most distinct errors is key to providing meaningful feedback to the students.
Nevertheless, it is not a comprehensive solution, as it cannot assess all categories of SQL errors or features. One might debate whether these limitations are the most critical
challenges encountered by SQL learners. Taipulus et al. \cite{taipalus2018} argue that, as teachers, “we should focus primarily on the errors [...] which affect the result table, and secondarily on complications, which affect performance” (or readibility). Ultimately, it is up to the instructor to determine whether the limitations of our model are acceptable in the context of his or her teaching.

\subsection{Implementation} \label{sec:implementation}

Now that we have covered the theoretical fundations of the mechanism, let us move on to its practical implementation. First, we will explain how to populate the \texttt{hash} columns of the core tables. Then, we will revisit the six functions introduced earlier and suggest at least one built-in or user-defined function for each. These functions are divided into two groups: those that do not involve aggregation are relatively straightforward to implement, while the aggregation functions require a more in-depth analysis. 

\subsubsection{Row hashes}

The row hash function $h$ takes as input the name of the table and its column list (excluding, obviously, the \texttt{hash} column itself). It serializes these values into a JSON array, which is then cast into a string and hashed with a UDF \verb|string_hash()|. Note that these calculations are not performed by the application layer, but by SQL triggers or DDL's generated columns. Delegating them to the database is key to support DML commands: as soon as a task involves inserting or updating rows, their hashes must be created or updated dynamically.

\subsubsection{Non-aggregation functions}

\paragraph*{Coalescing function $f^*$.} In SQL, the polymorphic function \verb|coalesce()| returns the first non-NULL value among its arguments. For convenience, we define around it a monomorphic, one-argument wrapper \lstinline|nn()|, which takes a nullable big integer and returns a non-NULL one. Pseudo-code: \verb|nn(x: bigint) -> bigint := coalesce(x, 42)|.

\paragraph*{Variadic function $f_v$.} The few variadic functions offered by SQL are either non-symmetric (\verb|coalesce()|, \verb|concat()|), or non-combining (\verb|least()|, \verb|greatest()|). Aggregation functions like \verb|sum()| are semantically suitable, but syntactically non-variadic. Some DBMSs, like PostgreSQL, would allow defining their variadic equivalents, but it is simpler to unfold their arguments into a short expression with a binary operator. In SQL, besides the standard addition (\verb|+|) and multiplication (\verb|*|), the bitwise operators AND (\verb'&'), OR (\verb'|'), and XOR (\verb'^') are commutative and supported by most implementations, including SQL Server, PostgreSQL, and MySQL\footnote{
    SQLite requires writing \verb'x ^ y' as \verb'(x | y) - (x & y)'.
}. For reasons that will be detailed in Section~\ref{sec:aggregate-partial}, we have opted for the addition.

\paragraph*{Control function $f_c$.} Similarly, the addition is a good candidate for the control function. This requires the control value to be a number. In other cases, we first convert it (if needed) into a string, and then into a big integer with our UDF \verb|string_hash()| (again).

\paragraph*{Salt functions $f_s$.} SQLab builds in the final database up to ten hundreds functions \texttt{salt\_ddd(x: bigint) -> bigint := nn(x) XOR Y}, where \texttt{ddd} is a three-digit number (also corresponding to the task number) and \texttt{Y} a random big integer constant.

At first glance, the coalescing operation \texttt{nn()} seems redundant with the initial application of $f^*$. However, consider a WHERE clause using a direct comparison with NULL instead of \lstinline|IS [NOT] NULL| (error 9 in \cite{brass2006}). Typically, such a query will return no rows, hence no token. Yet, when it begins with \SELECT $\mathcal{A}$\verb|(...)|, where $\mathcal{A}$ is an aggregation function, it will return exactly one row, and thus one token.

This token is calculated by Formula~\ref{eqn:formulas:agg-post}. First, $f_v(\{ f^*(h) \mid h \in \text{hashes} \})$ evaluates to NULL. Then, for most\footnote{Except \texttt{count()} and, in MySQL, \texttt{bit\_and()}, \texttt{bit\_or()} and \texttt{bit\_xor()}, cf. \href{https://dev.mysql.com/doc/refman/8.4/en/bit-functions.html}{\cite{sqlite}\texttt{/bit-functions.html}}.} aggregation functions, $f_a(\text{\NULL})$ is NULL. The next function, $f_w$, is also an aggregation, with the same outcome. Hence, $f_c$ and $f_s$ are our last chances for transforming a NULL argument into a non-NULL token. $f_c$ being user-facing, it is better to “hide” the coalescing operation \texttt{nn()} in the UDF $f_s$.

\subsubsection{Aggregation functions on full-range hashes}\label{sec:aggregate-full}

Functions $f_w$ and $f_a$ need to be chosen with care. Concerning $f_w$, as it is used in conjunction with an empty \lstinline|OVER ()| to fill the whole \texttt{token} column, we can rule out those window functions which only make sense when the result table is partitioned, i.e., ranking functions like \texttt{rank()}, \texttt{row\_number()}, \texttt{cume\_dist()}, \texttt{ntile()}, etc., and so-called value window functions like \texttt{first\_value()}, \texttt{nth\_value()}, \texttt{lag()}, etc. This leaves us with the same aggregation functions\footnote{
    Except if the DBMS forbids their use as window functions, like PostgreSQL for \texttt{percentile\_cont()} \cite{postgresql}.
} that are candidates for $f_a$.

\begin{table}[!ht]
    \small
    \tablecaption{Availability of the numeric aggregation functions considered for $f_w$ and $f_a$.}
    \begin{tabular}{
  >{}m{11.5em}
  >{\raggedright\arraybackslash}m{18.9em}
  >{\centering\arraybackslash}m{0.7em}
  >{\centering\arraybackslash}m{0.7em}
  >{\centering\arraybackslash}m{0.7em}
  >{\centering\arraybackslash}m{0.7em}
  >{\centering\arraybackslash}m{0.7em}
  >{\centering\arraybackslash}m{0.7em}
  >{\centering\arraybackslash}m{0.7em}
  >{\centering\arraybackslash}m{0.7em}
  >{\centering\arraybackslash}m{0.7em}}
  \makecell[l]{\\[1.5ex]\textrm{\bf Names}}
  & \makecell[l]{\\[1.5ex]\textbf{Comment}}
  & \rlap{\begin{turn}{60}\textbf{DuckDB}\end{turn}}
  & \rlap{\begin{turn}{60}\textbf{IBM Db2}\end{turn}}
  & \rlap{\begin{turn}{60}\textbf{MySQL}\footnotemark{}\end{turn}}
  & \rlap{\begin{turn}{60}\textbf{Oracle}\end{turn}}
  & \rlap{\begin{turn}{60}\textbf{PostgreSQL}\end{turn}}
  & \rlap{\begin{turn}{60}\textbf{Snowflake}\end{turn}}
  & \rlap{\begin{turn}{60}\textbf{SQL Server}\end{turn}}
  & \rlap{\begin{turn}{60}\textbf{SQLite}\footnotemark{}\end{turn}}
  \\[-3ex]
  \midrule
  {\footnotesize\verb|avg, count, max, min, sum|} & Core aggregate functions, already in the first version of the standard ISO 9075:1987. & \cmark & \cmark & \cmark & \cmark & \cmark & \cmark & \cmark & \cmark \\
  \hline
  {\footnotesize\verb|modular_sum|} & Wraps around the token's range. & ~ & ~ & ~ & ~ & ~ & ~ & ~ & ~ \\ 
  \midrule
  {\footnotesize\verb|percentile_cont|} & First (\verb|quartile_1|), second (\texttt{median}), third quartile (\verb|quartile_3|), and interquartile range (\texttt{iqr} $= Q3 - Q1$). & \cmark & \cmark & ~ & \cmark & \cmark & \cmark & \cmark & \gcmark \\
  \midrule
  {\footnotesize\verb|bit_and, bit_or, bit_xor|} & Bitwise aggregations. & \cmark & ~ & \cmark & ~ & \cmark & \cmark & ~ & ~ \\
  \midrule
  {\footnotesize\verb|checksum_agg (hash_agg)|} & 64-bit hash over the \emph{unordered} values. & ~ & ~ & ~ & ~ & ~ & \cmark & \cmark & ~ \\
\end{tabular}

    \label{tab:aggregate_functions}
\end{table}
\footnotetext{MySQL 8.4 doesn't support \texttt{percentile\_cont()}, cf. \url{https://bugs.mysql.com/bug.php?id=93234}.}
\footnotetext{In SQLite, \texttt{percentile\_xx} requires the \texttt{stats} third-party extension \cite{sqlean}.}

Among them, we choose to study those of Table~\ref{tab:aggregate_functions}, regardless of their availability in this or that DBMS. Suitable functions must satisfy the global properties of the fingerprinting mechanism listed in Section~\ref{sec:global-properties}.

\begin{itemize}
\item Their outcome must be an integer: float-valued functions like \texttt{avg()} are cast accordingly.
\item This integer must fit into a 64-bit range: the function \texttt{sum()}, notably, is at risk of overflow. Two solutions are discussed: in the present section, we replace \texttt{sum()} with its wraparound version, \texttt{modular\_sum()}; in Section~\ref{sec:aggregate-partial}, we will keep \texttt{sum()}, but use fewer bits for the initial hash values.
\item The functions must have the fewest collisions possible. This may disqualify the non-combining ones, i.e., those whose outcome depends only on a certain portion of the input data, like \texttt{max()} and \texttt{min()}. The functions \texttt{count()}, \texttt{bit\_or()}, and \texttt{bit\_and()} seem also problematic, as they only capture a small amount of information from their input.
\end{itemize}

To go beyond these intuitions, we conduct a series of tests comparing the outcome density of the selected functions. Differences are magnified by working on 16-bit instead of 64-bit integers. We first generate 900 individuals consisting of random sets of up to 100 hash values. To simulate the coalescing of NULL values, 50~\% elements of 10~\% of them are replaced by an arbitrary fallback constant equal to 25~\% of the hash upper bound (16384). To simulate the queries giving \emph{almost} the expected answer, 100 variants of selected individuals are created in \emph{all} the following ways: removing one random element, removing two random elements, flipping one random bit in one random element, flipping two random bits in two random elements. We then apply the integer version of each aggregation function to the resulting 1000 individuals. The outcomes are plotted in Fig.~\ref{fig:densities_full}.

\begin{figure}[!ht]
    \includegraphics[width=0.9\linewidth]{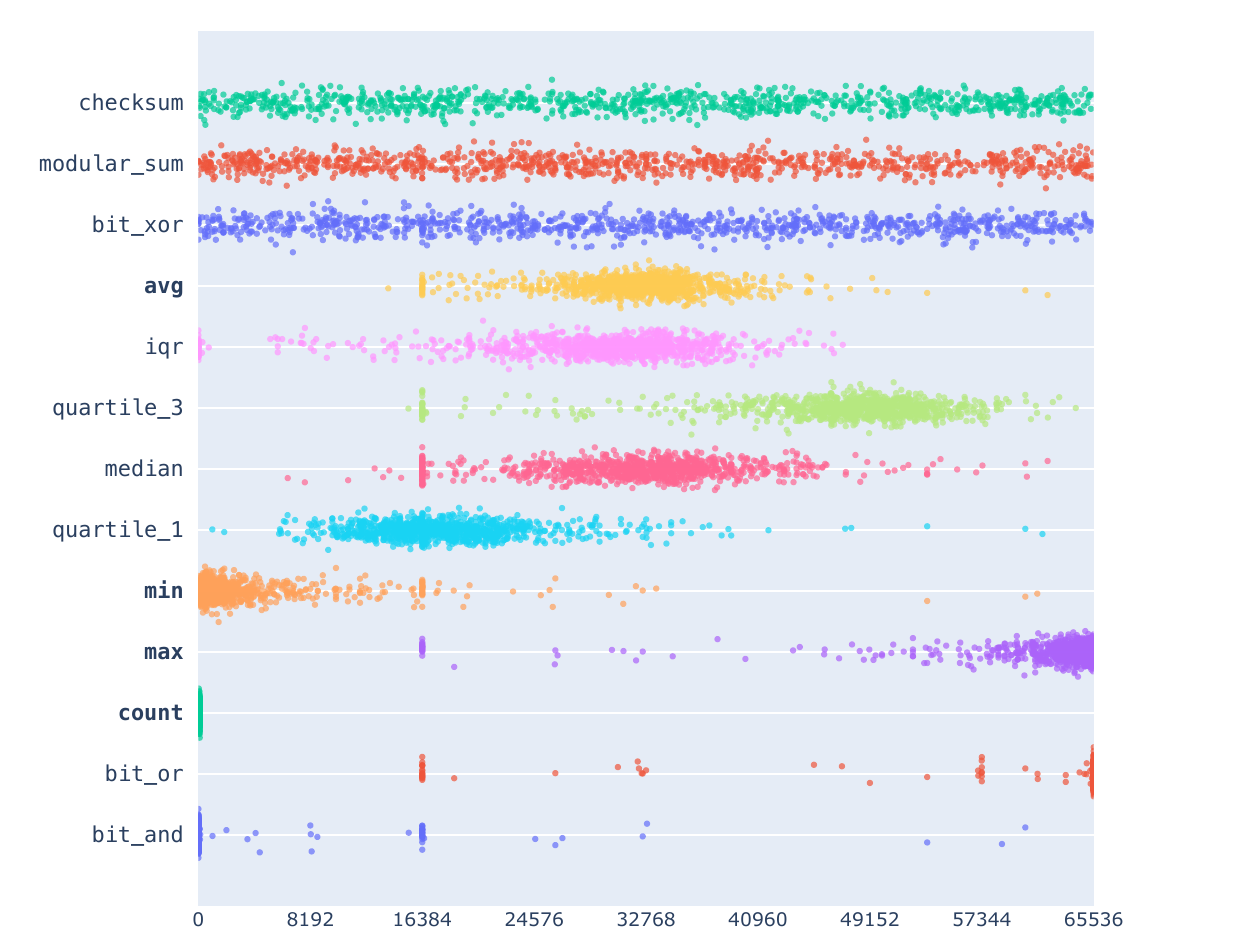}
    \caption{Density of aggregation function outcomes on full-range hashes (sparser is better). Each jittered strip plot represents the outcome of a function applied to 1000 individuals. \hfill\accmat}
    \Description{
        The figure shows density plots of various aggregation functions over a range of values from 0 to 65536. The x-axis spans this range. The aggregation functions and their densities are as follows:
        - most of the results of bit_and and bit_or fall at the ends of the range, with some outliers;
        - count appears as a vertical line at the start of the range.
        - min and max tend to concentrate at the ends of the range;
        - quartile_1, median and quartile_3 cluster around their respective quartiles, iqr and avg around the average value;
        - bit_xor, modular_sum and checksum are regularly distributed on the whole range.
    }
    \label{fig:densities_full}
\end{figure}

As expected, \texttt{bit\_and()}, \texttt{bit\_or()}, \texttt{count()}, \texttt{max()} and \texttt{min()} cluster against the bounds of the token's range. For the statistics functions, the clusters are wider and open on both ends. Note the sub-clusters around the fallback constant 16384. The remaining top three functions, \texttt{bit\_xor()}, \texttt{modular\_sum()}, and \texttt{checksum\_agg()}, exhibit a more uniform distribution.

\begin{figure}[!ht]
    \includegraphics[width=0.9\linewidth,trim={0cm 0.1cm 0cm 0cm},clip]{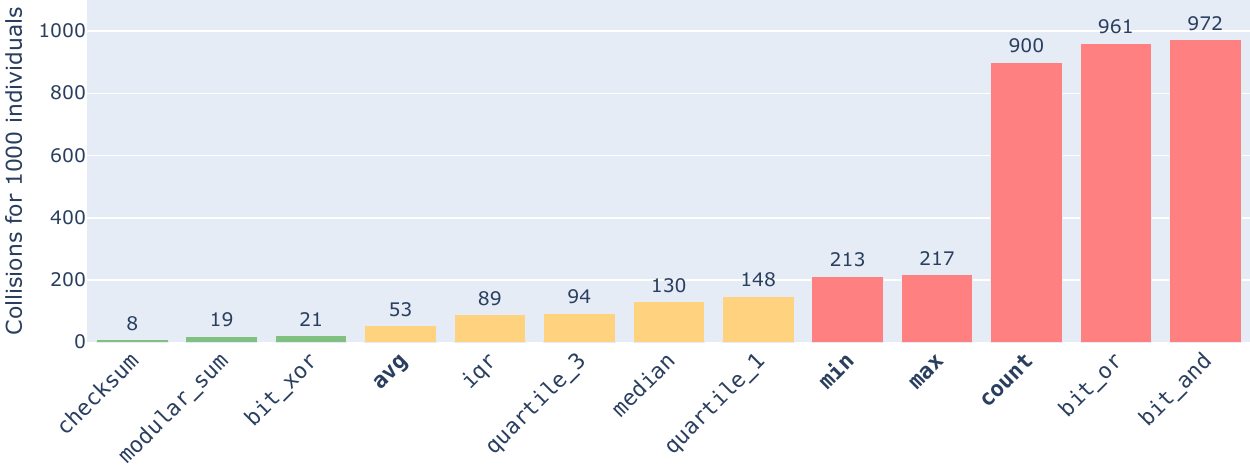}
    \caption{Collisions in aggregation function outcomes on full-range hashes (lower is better).
    In bold, the core SQL aggregation functions.
    In green [resp., orange, red], the functions that are suitable [resp. tolerable, unsuitable] for our purpose. \hfill\accmat
    }
    \Description{
        The histogram shows the aggregation functions by increasing number of collisions for 1000 individuals: checksum, modular_sum and bit_xor are below 25; avg, iqr, quartile_3, median and quartile_1 are between 50 and 150; min and max are just above 200; count, bit_or and bit_and are above 900.
    } 
    \label{fig:collisions_full}
\end{figure}

Fig.~\ref{fig:collisions_full} shows the collision count for each function. Obviously, \texttt{checksum\_agg()} would be ideal for our purposes. However, the fact that it is not widespread across DBMSs (cf. Table~\ref{tab:aggregate_functions}) is a major drawback when support for user-defined aggregation functions (UDAF) is lacking or would complicate the distribution of our SQLab database. For instance, PostgreSQL and MariaDB offer a CREATE AGGREGATE command which allows defining a new aggregation function in pure SQL, but MySQL and SQLite require writing, compiling and distributing a C/C++ shared library, which comes with is own challenges. Remember that we aim to provide a self-contained database that won't require any further installation by the end user.

Our ad-hoc function \texttt{modular\_sum()} is typical of this very issue. Let us see how we can avoid it altogether.

\subsubsection{Aggregation functions on partial-range hashes}\label{sec:aggregate-partial}

In this section, we will discuss how to “save” the built-in \texttt{sum()} function by preventing it from overflowing. The idea is to restrict the initial hash values to 40 bits. Thus, overflowing the 64-bit range would require at least $2^{64-40} = 16,777,216$ additions. Given the modest size of the datasets we work on (less than 50 rows for SQLab Island, less than 4000 rows for SQLab Sessform), this is rather unlikely\footnote{In our classes, it seems that only one student out of two hundreds has ever come across a sum overflow error.}. Let us now examine which impact this restriction would have on the other aggregation functions.

We conduct the same tests as in the previous section, with the following scale changes: 10-bit hash values and 16-bit tokens are used (the 10/16 ratio is the same as 40/64). We generate 500 individuals instead of 1000. The size of each individual is limited to 50 hash values instead of 100. These adaptations are designed to optimize the range of possible outcomes by ensuring that the maximum number of observed collisions remains just below the population size. The density plot is given in the accompanying notebook, and the collision histogram in Fig.~\ref{fig:collisions_partial}. \hfill\accmat

\begin{figure}[!ht]
    \includegraphics[width=0.9\linewidth,trim={0cm 0.1cm 0cm 0cm},clip]{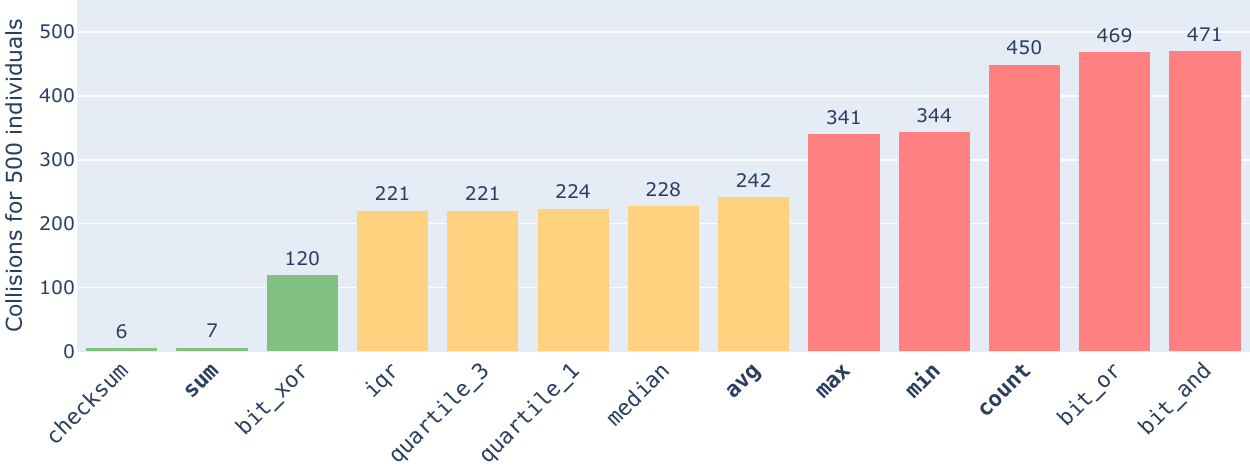}
    \caption{Collisions in aggregation function outcomes on partial-range hashes (lower is better). \hfill\accmat}
    \Description{
        The histogram shows the aggregation functions by increasing number of collisions for 500 individuals: checksum and sum are below 10; bit_xor is at 120; iqr, quartile_3, quartile_1, median and avg are between 220 and 250; min and max are just above 340; count, bit_or and bit_and are above 450. 
    } 
    \label{fig:collisions_partial}
\end{figure}

Compared with \texttt{modular\_sum()}, the behavior of \texttt{sum()} remains satisfactory, but at the cost of a deterioration in that of subsequent functions. Remember, however, that the actual hash range is on 40 bits (1,1 trillion possible values), not 10 bits (1024 possible values): in practice, the rate of collisions of any reasonably input-sensitive function (i.e., those on the left of \texttt{max()} and \texttt{min()} in Fig.~\ref{fig:collisions_partial}) will be vanishingly low.

\subsubsection{Composition of the aggregation functions}

One detail remains to be settled, which involves the associativity property. Associativity, first identified in the context of binary operations such as addition and multiplication, easily and crucially\footnote{
    Parallel and distributed computing can take advantage of the ability to group elements without affecting the outcome.
} extends to aggregation functions. Several equivalent definitions are found in the literature, 
for instance in \cite{marichal2015}: an aggregation function $\mathcal{A}: E^n \rightarrow E$ is \textbf{associative} if and only if
\begin{align*}
    \label{eqn:associativity}
    &\forall x \in E,  \mathcal{A}(x) = x\\
    &\forall x_i \in E, \mathcal{A}(x_1, ..., x_j, x_{j+1}, ..., x_k, x_{k+1}, ..., x_{n}) =
\mathcal{A}(x_1, ..., x_j, \mathcal{A}(x_{j+1}, ..., x_k), x_{k+1}, ..., x_{n}).
\end{align*}

Among the functions of Fig.~\ref{fig:collisions_partial}, only
\texttt{sum()},
\texttt{bit\_xor()},
\texttt{max()},
\texttt{min()},
\texttt{bit\_or()}, and
\texttt{bit\_and()}
are associative. For instance:
$
    \texttt{max}(1,2,3,4,5,6) = \texttt{max}(1,2,\texttt{max}(3,4,5),6)
                              = \texttt{max}(1,2,5,6)
                              = 6
$.

It follows that, for any symmetric associative aggregation function $\mathcal{A}$, any input multiset $X$ and any two partitions $(X_1, ..., X_p)$ and $(X'_1, ..., X'_p)$ of  $X$, we have
$
\mathcal{A}(\mathcal{A}(X_1), ..., \mathcal{A}(X_p)) = \mathcal{A}(\mathcal{A}(X'_1), ..., \mathcal{A}(X'_p)) = \mathcal{A}(X)
$.
Therefore, our two aggregation functions $f_w$ and $f_a$ cannot be independently chosen: if both are the same associative function, all groupings result in the same fingerprint (Fig.~\ref{ex:group_blindness}).

\begin{figure}[!ht]
    \begin{tabular}{V{0.31\textwidth}V{0.31\textwidth}V{0.31\textwidth}}
\begin{lstlisting}
SELECT count(*),
  sum(sum(hash)) over ()
FROM dependent
GROUP BY sex
\end{lstlisting}
&
\begin{lstlisting}
SELECT count(*),
  sum(sum(hash)) over ()
FROM dependent
GROUP BY relationship
\end{lstlisting}
&
\begin{lstlisting}
SELECT count(*),
  sum(sum(hash)) over ()
FROM dependent
\end{lstlisting}
\\
\begin{verbatim}
    4 | 4064302514963
    3 | 4064302514963
\end{verbatim}
&
\begin{verbatim}
    2 | 4064302514963
    3 | 4064302514963
    2 | 4064302514963
\end{verbatim}
&
\begin{verbatim}
    7 | 4064302514963
\end{verbatim}
\end{tabular}

    \caption{Identical 4064302514963 fingerprint for three queries with different or even non-existent GROUP BY clause. “Group blindness” is a consequence of choosing the same associative aggregation function for both $f_w$ and $f_a$. Token formulas and result tables simplified for clarity. \hfill\accmat}
    \label{ex:group_blindness}
    \Description{Text.}
\end{figure}

Besides this rule, one can theoretically choose any associative aggregation function(s) for $f_w$ and $f_a$, except the couple (\texttt{sum()}, \texttt{count()}) which exhibits an “inter-associativity” property:
$
    \texttt{sum}(\texttt{count}(1,2),\texttt{count}(3,4,5)) = \texttt{sum}(2,3)
    = 5
    = \texttt{sum}(4,1)
    = \texttt{sum}(\texttt{count}(1,2,3,4),\texttt{count}(5))
$.

\subsubsection{Final recommendations}

For the first required aggregation function, take the leftmost available function of Fig.~\ref{fig:collisions_partial}. If this function is non-associative, take it for the second one as well. Otherwise, take the next one. Should one of them be \texttt{bit\_xor()}, avoid using it for $f_a$: it would reduce duplicate hash values $x$ to either zero or $x$, which may result in fingerprint collision (cf. Annexes, Table~\ref{tab:xor_sum_vs_sum_xor}). Table~\ref{tab:fw_fa} gives the outcome of this process for some DBMSs. Of course, when distributing a user-defined aggregation function with the database is possible, it is better to create a \texttt{checksum\_agg()} implementation for both $f_w$ and $f_a$.

\begin{table}[!ht]
    \tablecaption{Recommended built-in implementations of $f_w$ and $f_a$ across various DBMSs. If \texttt{sum()} is used, remember to limit the hash size to 40 bits to avoid overflow.}
    \begin{tabular}{lll}
    ~
    & $f_w()$
    & $f_w(f_a())$ \\
    \cmidrule(lr){2-3}
    SQL Server & \texttt{checksum\_agg()} & \texttt{checksum\_agg(checksum\_agg())} \\
    Snowflake & \texttt{hash\_agg()} & \texttt{hash\_agg(hash\_agg())} \\
    PostgreSQL, MySQL, DuckDB & \texttt{sum()} & \texttt{bit\_xor(sum())} \\
    SQLite, Oracle, IBM Db2 & \texttt{sum()} & \texttt{sum(avg())}
    \end{tabular}
    \label{tab:fw_fa}
\end{table}

By default, SQLab uses 40-bit partial-range hashes. The recommended functions are those that had the fewest collisions in the tests we conducted, but they are not guaranteed to fit all queries on all datasets. A warning will be emitted if any collision occurs during the build process. In this event, the designer may choose another function for $f_w$ and/or $f_a$ in the offending formula.

As a concluding example, here is the full token formula again, followed by a possible MySQL implementation when three tables are involved in the outer query:
\begin{gather}
f_s\bigg(
    f_c\Big(x, 
        f_w\big(
            f_a\big(
                f_v\big(
                    \{ f^*(h) \mid h \in \text{hashes} \}
                \big)
              \big)
        \big)
    \Big)
\bigg)
\tag{\ref{eqn:formulas:agg-post}}
\\
\texttt{salt\_042(\{\{x\}\} + bit\_xor(sum(nn(A.hash) + nn(B.hash) + nn(C.hash))) OVER ()) as token}
\tag*{}
\end{gather}

\section{Life cycle of an SQLab database} \label{sec:life_cycle}

\subsection{Design process} \label{sec:design_process}

\subsubsection{Source material}\label{sec:source-material}

To create a database of exercises or adventures, the designer provides three assets:

\noindent\begin{minipage}{\textwidth}
\setlength\intextsep{-\baselineskip}
\begin{wrapfigure}{R}{0pt}
\begin{lstlisting}
    §~§
    §{\textcolor{lightgray}{ddl.sql (extract)}}§
  §\textcolor{lightgray}{\rule[1ex]{0.48\textwidth}{0.5pt}}§
  CREATE TABLE department (
      dpt_name VARCHAR(15) NOT NULL UNIQUE,
      dpt_id INT NOT NULL PRIMARY KEY,
      manager_id CHAR(9) NOT NULL,
      manager_start DATE,
      hash BIGINT
  );

  §{\textcolor{lightgray}{dataset/department.tsv}}§
  §\textcolor{lightgray}{\rule[1ex]{0.48\textwidth}{0.5pt}}§
  Research      §\textcolor{lightgray}{$\rightarrow$}§5§\textcolor{lightgray}{$\rightarrow$}§333445555§\textcolor{lightgray}{$\rightarrow$}§1988-05-22
  Administration§\textcolor{lightgray}{$\rightarrow$}§4§\textcolor{lightgray}{$\rightarrow$}§987654321§\textcolor{lightgray}{$\rightarrow$}§1995-01-01
  Headquarters  §\textcolor{lightgray}{$\rightarrow$}§1§\textcolor{lightgray}{$\rightarrow$}§888665555§\textcolor{lightgray}{$\rightarrow$}§1981-06-19
  §~§
\end{lstlisting}
\end{wrapfigure}
~
\begin{enumerate}
    \item The definition of the core dataset structure, consisting in a DDL script in the SQL dialect of the target DBMS (MySQL on the right). Each table must contain an additional \texttt{hash} column.
    \item A directory of tab-separated values files, each containing the data for one table. The \texttt{hash} column is not provided by the designer: it will be generated or calculated by a trigger.
    \item A \textbf{master notebook} specifying a set of independant or interdependent tasks. A task typically consists of a narrative, a problem statement, a primary gold query, and any number of variants and hints. The document is created in the Jupyter Notebook format \cite{kluyver2016}, which combines text cells (Markdown) and code cells (Python, SQL, or other languages) that can be executed in any order.
\end{enumerate}
\end{minipage}\\

\begin{figure}[!ht]
    \includegraphics[trim={2.55cm 10.5cm 5cm 3.55cm},clip]{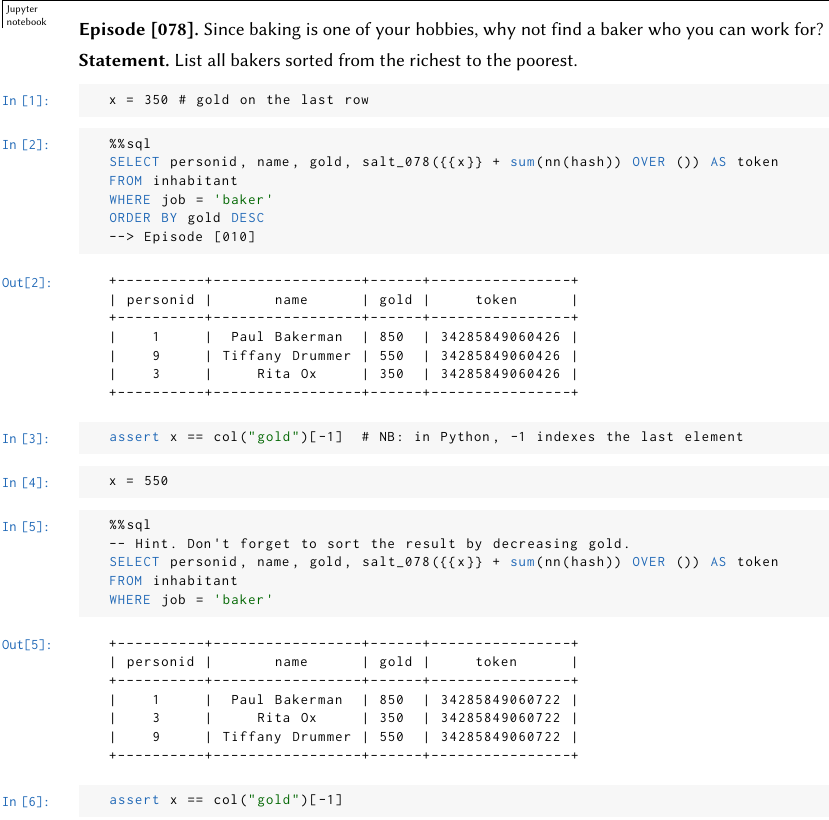}
    \caption{An extract of the master notebook of SQLab Island \cite{sqlab_island}, our adaptation of SQL Island \cite{schildgen2014}.}
    \label{fig:notebook}
    \Description{Text.}
\end{figure}

\subsubsection{Master notebook} The main inspiration for the master notebook syntax stems from the stated objectives of the Markdown format \cite{gruber2004}: “make it as readable as possible [...] without looking like it's been marked up with tags or formatting instructions”. Ideally, the reader should forget that the master notebook is destined to be parsed into a database. In this spirit, plain-english sectionning names (e.g., “Episode”, “Statement”, “Hint”) double as keywords. Furthermore, given that a notebook is executable, most of the machinery is delegated to the interpreters: SQL evaluates the queries and the tokens, Python checks the expected properties and control values, and the designer can focus on the educational and narrative content.

The example of Fig.~\ref{fig:notebook} defines a task pertaining to the SQLab Island adventure \cite{sqlab_island}. Let us dissect it:
\begin{itemize}
\item First, a Markdown cell introduces the task by specifying its type, its number and a bit of text (narrative, explanation, etc.). In an introductive cell, two types are possible: “Episode” (for an adventure) and “Exercise” (for an independant task).
\item The second cell, also in Markdown, provides the main statement of the task (type “Statement”).
\item The next cell is mandatory when the task involves a control value. It consists in a Python assignment of the form \verb!x = value # comment text!. The comment text will be injected in the instruction template “replace \texttt{(0.0)} with \verb!{}.!”, as shown in Fig.~\ref{ex:fingerprinting_2}.
\item The input part of the following cell starts by \verb!%%sql!, a “cell magic” of the Jupyter extension JupySQL \cite{jupysql}. It specifies the primary gold query, i.e., the expected answer to the problem statement, with its token formula.
\begin{itemize}
    \item Note that this formula uses the function \verb!salt_078()!, where \texttt{078} is also the task number.
    \item The expression \verb!{{x}}! refers to the two-pass fingerprinting control value in the syntax of JupySQL interpolation. It will be replaced with \texttt{350} during the execution, and with the placeholder \texttt{(0.0)} during the generation of the messages.
    \item The input cell ends with an SQL comment disguised as a long arrow \verb!-->!, followed by a reference to the episode that the token produced by this query would unlock.
\end{itemize}
\item The output part of the same cell is the result table, with its mandatory column \texttt{token}.
\item It is followed by an optional Python statement, which asserts that the last row in the \texttt{gold} column is indeed the expected control value. The helper function \texttt{col()} is used to access the values of a given column of the \emph{previous} result table.
\item Although this is not the case here, any number of secondary gold queries could follow. They will generally produce the same token value as the primary gold query. Otherwise, the token formula is mandatory. By default, all the gold queries will be concatenated to form the correction displayed to the students when they have unlocked a correct answer.
\item The three last cells apply the same logic to the specification of a hint. An incorrect query is introduced by a \verb!-- Hint! towards the correct answer. It must use the same formula, and produce a different token.
\end{itemize}

For more details on the structure and syntax of the master notebook, the reader is invited to consult the complete example \href{https://github.com/laowantong/sqlab_island/blob/main/en/sqlite/sqlab_island.ipynb}{\cite{sql_island_repo}\texttt{/blob/main/en/sqlite/sqlab\_island.ipynb}} or the documentation of SQLab \cite{sqlab}.

\subsubsection{Design notes}

\paragraph*{Coverage.} The database instance should kill all non-equivalent mutants, i.e., make all incorrect queries produce a token distinct from the correct ones: this is \textbf{full coverage}. Additionally, every incorrect query should be associated with a different hint, i.e., produce a token distinct from the incorrect ones: this is \textbf{discriminating coverage}. As we have discussed in Sections~\ref{sec:out-of-scope} and \ref{sec:false-positives}, these properties are not always possible to achieve. When they are, it is the duty of the designer to ascertain whether an update to the core dataset is required.

\paragraph*{Testing.}
Given that the core dataset is likely to evolve throughout the fine-tuning of the set of tasks, it is important to ensure that no regression occurs. To this end, we strongly recommend testing the output of each SQL cell. Two examples of such tests have been given in the context of two-pass fingerprinting (Fig.~\ref{fig:notebook}, cells 3 and 6). One-pass fingerprinting is no less amenable to testing. For instance, in the same game, the gold query for the question “Find the name of the chief of the village Onionville” is followed by:
\lstinline[language=Python]|assert col("name") == ["Fred Dix"]|.
Two key properties are checked here: 1. the result table contains a single row, and 2. the name of the chief is the expected one.

\paragraph*{Choice of the formula.}
A given query cannot be compatible with both a basic and an aggregation formula. In the case of a query with vector aggregate, one can hesitate between Formula~\ref{eqn:formulas:agg} and its controlled version~\ref{eqn:formulas:agg-post}, see Section~\ref{sec:main-scope} for a brief discussion. 
Concerning the dimension of the formula, it is generally dictated by the number of tables involved in the outer FROM clause of the gold query. Reducing the dimension may increase the stylistic freedom of the students, but at the cost of a higher risk of false positives (see below, first item).

\paragraph*{Concluding example.}
Some of these notes, as well as the observations made in Section~\ref{sec:accuracy}, could lead to the following advices for the design of the task of Table~\ref{tab:same_result_queries_basic}:
\begin{enumerate}
    \item opt for a two-dimensional instance of the formula. The one-dimensional formula allows for 5 gold queries out of 5, but detects only 1 out of 8 complications and 4 out of 10 errors. In contrast, the two-dimensional formula forbids the gold query $G_1$, but detects 2 more complications and 3 more errors.
    \item make the JOIN solution $G_2$ the first gold query;
    \item skip its symmetric $G_3$: as a general rule, for a query involving a $n$-dimensional formula in its outer SELECT clause and $n$ tables in its outer FROM clause, only one permutation is needed;
    \item mark the “inaccessible” one-dimensional query $G_1$ as a variant. Variants are not required to include the formula, and are simply presented to the students when they have unlocked a correct answer;
    \item mark the old-fashioned $G_4$ as another variant, with a brief explanation of the deprecated join syntax;
    \item for $O_3$, $O_7$ and $O_8$, provide a hint that only two tables are needed;
    \item modify the \texttt{employee} or \texttt{works\_on} table (either by inserting chaff or by updating the existing data) to make $I_1$, $I_2$ and $I_8$ yield a non-gold token, then provide the appropriate hints;
    \item for $I_7$, $I_9$ and $I_{10}$, provide a hint that the join condition is missing or incorrect.
\end{enumerate}

\subsection{Build process}

SQLab provides a subcommand \texttt{create} to build the database from the source material. It is a five-step process:

\begin{enumerate}
\item The database is created or recreated from the DDL script.
\item It is populated with the data from the TSV files and the corresponding hash values.
\item The master notebook is executed, and all the code cell outputs are updated.
\item It is parsed into an intermediate linear representation, the \textbf{records} (a JSON list of objects).
\item The records are traversed to populate a single-column table of encrypted messages, called \texttt{sqlab\_msg}.
\end{enumerate}

Throughout this process, the following properties are checked:

\begin{enumerate}
\item \emph{Hash uniqueness.} No hash value is duplicated within a table or across several tables\footnote{The latter explains why the (weaker) constraint \texttt{UNIQUE} was omitted from the \texttt{hash} column of the DDL file shown in Section~\ref{sec:source-material}.}.
\item \emph{Validity.} No SQL or Python code cell raises an error.
\item \emph{Structural conformity.} The cells of different natures follow a predefined order, for instance: episode or exercise header, context (optional), statement, primary solution, variants (optional), hints (optional).
\item \emph{Consistency.} In a given formula, the salt function has the same number as the task to which it belongs.
\item \emph{Usefulness.} Any primary solution or hint must produce a token (otherwise it could not be accessed)\footnote{
    A variant is not required to produce a token (since it can be accessed from the token of its associated primary solution).
}.
\item \emph{Token uniqueness.} Primary solution or hint tokens must be unique. Variant tokens must be distinct from hint tokens, but not necessarily from their associated primary solution token, nor among themselves\footnote{
    A variant may produce a token distinct from the token of its associated primary solution. In this case, the two tokens are functionally equivalent (they give access to the same message).
}.
\item \emph{Round-trip.} A message encrypted with a specific token can be decrypted using the same token.
\end{enumerate}

\subsection{Practice phase} \label{sec:practice_phase}

Once an SQLab adventure is compiled into an SQL dump, it is time to let the students engage with the database.

\subsubsection{Undesirable outcomes of our model}
The fingerprinting mechanism is a powerful tool, but also a source of potential frustrations. The instructor shall initially explain and demonstrate its operational principles. The first two or three queries should allow the students to grasp the concept, with the real SQL difficulties being deliberately kept to a minimum.

\begin{itemize}
\item \emph{Formula errors.}
Inserting a token formula in the outer SELECT clause may “break” a working query. The students must be warned of this risk of errors, which are not present in traditional exercises. They can be classified into three categories:

\begin{itemize}
    \item \emph{Compatibility errors.} In strict mode, a query cannot be compatible with both a basic and an aggregation formula. Suppose that a student attempts to solve a deduplication task with a GROUP BY clause, while the instructor has devised it with a DISTINCT keyword. Evaluating the provided Formula~\ref{eqn:formulas:basic} will raise a runtime error which may be difficult to interpret. See Annex~\ref{sec:incompatibility_handling} for a discussion on the handling of such situations.
    \item \emph{Reference errors.} If the formula involves an alias that is not defined in the outer FROM clause, the query will fail with an “Unknown column” error. The case where an expected alias refers to a wrong (extraneous) table is less problematic, as it will produce tokens which can be associated with specific hints.
    \item \emph{Control value errors.} In a two-pass process, the following errors may occur:
    \begin{itemize}
        \item \emph{Missing control value}, i.e., the placeholder is left at its default value of 0.0.
        \item \emph{Obsolete control value}, i.e., the student forgets to update the result value from a previous attempt.
        \item \emph{Misidentified control value}, e.g., the student is instructed to copy-paste the third value of a column, but copy-paste the fourth instead.
    \end{itemize}
    Note that these errors do not prevent the query from executing, meaning that the tokens they produce can be associated with a hint.
\end{itemize}

\item \emph{Unintentional clues and anti-clues.}
Strictly speaking, SQLab cannot create traditional blank-page exercises. Indeed, providing a token formula inevitably discloses two key features of the gold query: the minimal number of tables involved in the outer FROM clause, and the presence of an outer GROUP BY clause and/or an aggregation function in the outer SELECT clause. While the second clue may escape some students, the first quickly becomes obvious to all. This can be of some help, for example, in the case of a self-join, where realizing that the same table is involved twice is difficult for a majority of students, even those who have no problem with Cartesian products or natural joins \cite{miedema2022b}.

If the designer wishes, he or she can provide additional clues by naming the tables in the formula with the initials of those involved in the FROM clause.
Not doing so forces the students to alias these tables with the letters A, B, C. This is a stylistic anti-pattern, as it reduces the readability of the query and may lead to swapped aliases \cite{miedema2022b}.
\end{itemize}

Note that these undesirable outcomes are not intrinsic to the fingerprinting principle itself; rather, they stem from its minimalistic implementation as a standalone database. Such an approach delegates to the user several manipulations. They could be readily automated if the database were integrated into a dedicated learning platform. This is entirely possible, as we will demonstrate in Section~\ref{sec:embedded-version}.

\subsubsection{Benefits of our model}

In theory, students should no longer need the instructor to assess the correctness of their answers, help them when they are stuck, and provide them with the answer keys or the next task.

However, a dedicated instructor might leverage the practical session to its fullest potential: students are arguably the ultimate fuzzing tools, and can yield a variety of misunderstandings, misconceptions, and errors that few teachers would have thought of. This spontaneous production represents a great opportunity to enrich the database with hints and clarifications that are almost guaranteed to be useful in the future.

Unfortunately, even if students are instructed to call their instructor every time they stumble across an unexpected token, few of them actually do. Automatically logging their queries and the tokens they pass to the \texttt{decrypt()} function constitutes a non-intrusive way to gather an unbiased feedback on their difficulties and their progress. SQLab has a sub-command \texttt{report} to parse the logs and surface the new errors. Hot updates of the database are possible, but unpredicted errors sometimes require close examination, and one must be prepared to \emph{invest} time in this task\,---\,for the greater benefit of the students to come.


\section{Conclusion} \label{sec:conclusion}

\subsection{Main contributions}

\begin{enumerate}
\item Sections~\ref{sec:fingerprinting_concepts} to \ref{sec:accuracy} present a novel approach to the problem of semantic equivalence of SQL queries, based on the concept of fingerprinting. Mathematically, it relies on a hash function and five distinct formulas composing up to six functions. Such a formula calculates a fingerprint of the “starred” version of a query, which can informally be thought of as its restriction to the clauses executed before any projection. We prove that comparing fingerprints is equivalent to performing execution matching on starred queries. This comparison results in fewer false positives than execution matching on the original queries.
\item Section~\ref{sec:implementation} proposes several adaptations of these formulas to the practical constraints of a DBMS, informed by a statistical study of the outcomes of a set of aggregation functions on both full-range and partial-range big integer hashes.
\item Section~\ref{sec:individual_tasks} demonstrates the use of these fingerprints as tokens enabling selective access to pertinent information and feedback in an additional table of independently encrypted messages.
\item In \cite{sqlab}, we leverage these concepts and techniques in a new game engine, SQLab, developed in Python as a command-line tool. Its open-source licensing scheme (MIT), along with the fact that the games created with SQLab are playable without SQLab, should guarantee its longevity, regardless of academic imperatives and budgets. It already supports the three most popular open-source RDBMSs \cite{db_engines}: MySQL, PostgreSQL and SQLite. For each one, a proof-of-concept SQLab adaptation \cite{sqlab_island} of the well-known SQL Island game \cite{schildgen2014} is distributed under CC-BY-SA-4.0 (a content license, since an SQLab game is nothing but a database).
\item An SQLab game can be \emph{directed}, \emph{standalone}, \emph{stateful} and provide \emph{specific feedback} to the player. Section~\ref{sec:sql_games} claims that no existing SQL game exhibits all four properties, with the first two previously regarded as mutually exclusive. 
\item Section~\ref{sec:design_process} describes the main components and structure of the source material of an SQLab game, including the collection of related or unrelated tasks that constitute an adventure or exercises, respectively. This collection, formatted as a Jupyter notebook, interweaves SQL queries and basic Python assignments with narrative, educational or instructional texts.
\item Section~\ref{sec:practice_phase} lists the benefits and setbacks that students and instructors may encounter during an SQLab practical session. It underscores the critical role of each party, and hints at the prospective advantages of an embedded version of the game, with more details to be provided in Section~\ref{sec:embedded-version}.

\end{enumerate}


\subsection{Future directions}

\subsubsection{Empirical study}

SQLab was developed and tested in an educational setting over a three-year period, with two cohorts of undergraduate and master's students participating in each iteration. The inaugural adventure was found to be promising, but excessively ambitious, with only two students out of 100 successfully completing it (and this was far beyond the allotted time). Moreover, corrective measures were deemed necessary in the mechanism's design. From there, a simpler, more gradual adventure \cite{sqlab_sessform} was offered, with significant enhancements to the model. In the third year, a series of independent exercises and a mock examination were incorporated, and the software underwent a substantial revision for its first public release. The subsequent phase will entail a statistical investigation of the tool's efficacy, employing server logs and a comparative analysis of the outcomes of students utilizing SQLab and those engaged in more traditional practice. This study will be the subject of a separate article.

\subsubsection{DBMS support}

Extending SQLab to other RDBMSs should be straightforward, as long as they provide the following “modern” SQL features: window functions, triggers, user-defined functions, one hash function, and one decryption function (see Annex~\ref{sec:dbms-requirements} for more details on these requirements). Each adaptation requires about one hundred lines of Python, and as many lines of SQL code (see for example 
\href{https://github.com/laowantong/sqlab/tree/main/sqlab/dbms/mysql}{\cite{sqlab}\texttt{/tree/main/sqlab/dbms/mysql}}).

\subsubsection{Design assistance}

Automatically providing the designer with the following elements should be easy:

\begin{enumerate}

\item \emph{Formulas.} Most of the time, for a given query, the appropriate formula can be inferred from the number of tables involved in the outer FROM clause and the presence of an outer GROUP BY clause and/or an aggregation function in the outer SELECT clause. Auto-completing the query with this formula would be convenient.

\item \emph{Assertions.}
Testing the result table of each query is a safety net against the regressions that are bound to occur when the dataset evolves. These assertions are often predictable. For instance, it is common to test the cardinality of the result, the content of a projected column, or the presence of the expected control value. It should be possible to automatically infer some of these potential assertions from the statement and/or the category of the formula.

\item \emph{Mutants.}
Currently, the game designer is tasked with predicting, testing, and documenting all incorrect queries. To lighten this burden, the most straightforward mistakes could be generated. As we have seen in Section~\ref{sec:mutation}, this is a known problem, and a variety of generic mutation operators have been proposed by the research community. In addition, token-specific operators could automate the handling of the cases where students forget to replace a placeholder with the actual control value, or where they reuse the formula of a task in the next one.
A number of mutants could come with a hint consisting in a templated feedback message (level 3 of 4). These would include inadvertent Cartesian products, use of a wrong table in a join, confusion between strict and non-strict inequalities, and so forth.

In an ideal scenario, any (incorrect) mutant that would survive the database instance should result in the proposal of new or updated rows. 

\end{enumerate}

\subsubsection{Embedded version} \label{sec:embedded-version}

As we have seen, the defining characteristics of SQLab games are their \emph{directed} and \emph{standalone} nature. The latter does not prevent them to be embedded in a dedicated learning and training platform. On the contrary, such a platform would bring a number of benefits that, arguably, could unlock the full potential of our model. Here are a few areas where it could literally be a “game-changer”:

\begin{enumerate}

\item \emph{Interface.} Command-line and administration tools are not the most user-friendly interfaces. A dedicated front-end could allow the designer to enrich the storyline with images, animations, sound effects, etc. Navigation, accessibility, and responsivity could be greatly improved. Feedback messages could be displayed in a more appealing way: currently, some Markdown typographical enrichments are simulated with Unicode to compensate for the fact that administration tools can only display tables in plain text (cf. Fig.~\ref{fig:pma}).

\item \emph{Fingerprinting mechanism.} It could be made entirely transparent to the students. Each time they would submit a query, the system would automatically inject the token formula, execute the augmented query, optionally update the formula with a required control value and re-execute the query, retrieve the final token, and finally display both the result table (without this token) and the decrypted feedback message. Similarly, “saving” a stateless game could be as simple as recording the last token that was successfully submitted.

\item \emph{Gamification.} The platform could include a scoring system, a leaderboard, badges, personalized certificates (like SQL Island \cite{sql_island_online}), etc. We are aware that game based learning is not necessarily more efficient than traditional teaching methods \cite{hattie2008}. However, gamification \cite{deterding2011} has been proved to increase the motivation and the engagement of the students \cite{bai2020,ren2024,tahir2020}.

\item \emph{Grading.} An embedded version of SQLab could also be used to turn a set of exercises into a formal exam, and to grade it automatically or semi-automatically (given the limitations listed in Section~\ref{sec:out-of-scope}). SQL query evaluation based on similarity measures is an active research area \cite{kaur2023,koberlein2024,wang2020,yang2022}. It offers two undeniable advantages: saving time for graders and ensuring fairness among students. However, in our view, complex calculations of edit distances or other metrics tend to evaluate a query as a collection of unrelated deviations rather than a coherent whole, whereas a human grader's approach is often more holistic. Is is easy to imagine how a grading assistance system based on our model's fingerprint mechanism could work. For a given query, SQLab would likely present to the grader the various answers, grouped by their fingerprint and ordered by decreasing frequency. Regardless of the number of students, the grader would likely face only a handful of distinct answers, which could be cross-compared and evaluated normally. All responses would thus be graded fairly. Then, a second review may evaluate the complete submissions of each student. This would avoid, for example, penalizing $n$ times $n$ repetitions of the same mistake by the same student.

\end{enumerate}

\begin{acks}
We would like to thank Prof. Dr. Johannes Schildgen for responding quickly and favorably to our request to place SQL Island under an open-source licence, and for granting us permission to adapt it to SQLab. Our gratitude extends to our colleagues Kamel Chelghoum and Imed Kacem for the motivating discussions, Dominique Roméo, Pascal Krier and Giorgio Lucarelli for their help in setting up and running the first SQLab practical sessions, and, last but not least, our students: without them, none of this would have been possible.
\end{acks}

\printbibliography[nottype=software,title={References}]
\printbibliography[type=software,title={Softwares}]

\appendix

\section{Annexes}

\subsection{Supplementary figures and tables}

\begin{figure}[!ht]
    \includegraphics[width=0.9\linewidth]{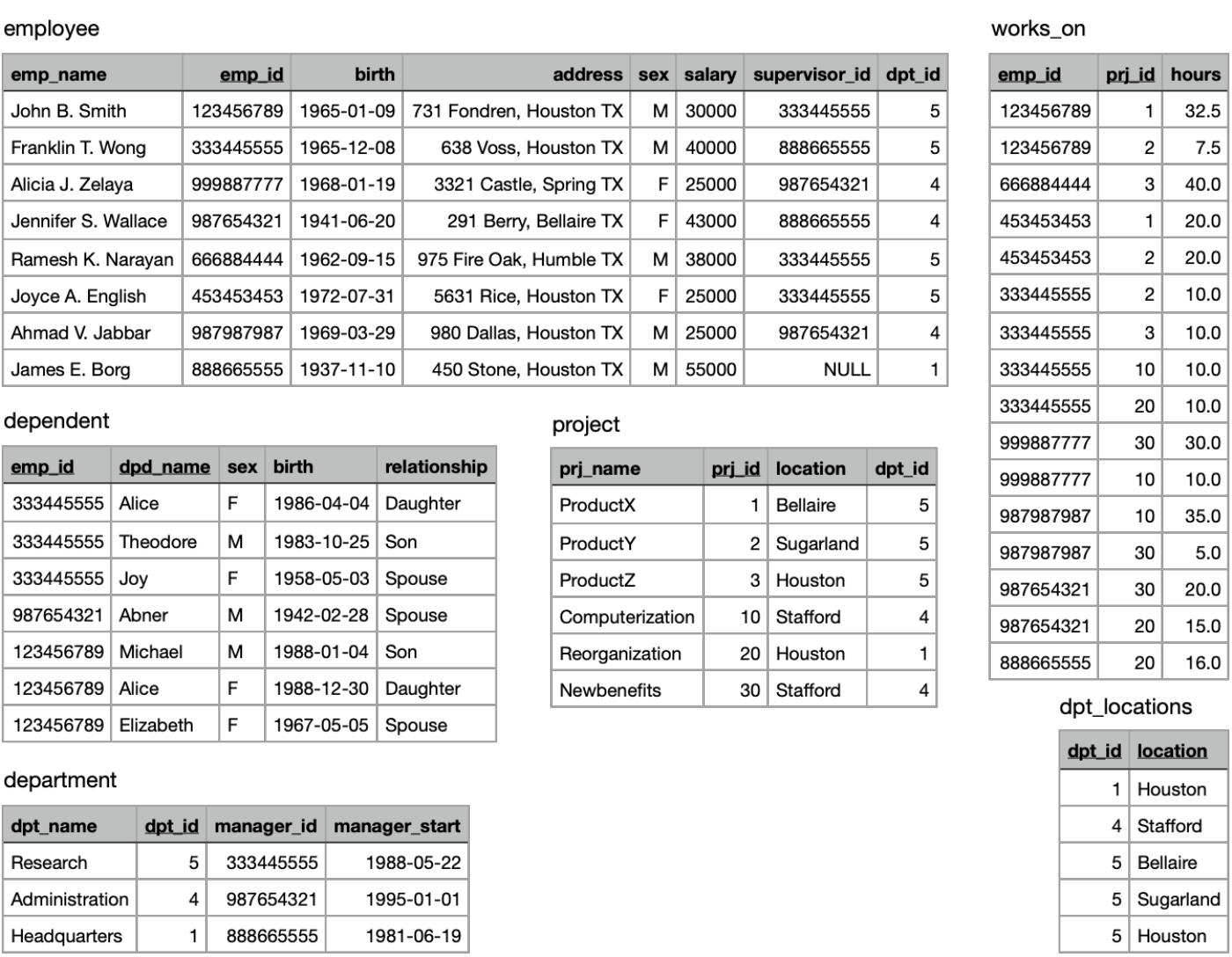}
    \caption{The majority of the queries studied in the present paper use the schema of the \texttt{company} database defined in the textbook of Elmasri and Navathe \cite{elmasri2017}.
    When provided, the result tables are calculated on the instance of p.~250 (ibid.).
    Note that some identifiers were renamed, and columns \texttt{Fname}, \texttt{Minit}, and \texttt{Lname} were merged into a single column \texttt{emp\_name}.}
    \label{fig:company}
    \Description{The company database consists of 6 tables. The employee table includes emp_name, emp_id, birth, address, sex, salary, supervisor_id, and dpt_id. The department table includes dpt_name, dpt_id, manager_id, and manager_start. The dpt_locations table includes dpt_id and location. The project table includes prj_name, prj_id, location, and dpt_id. The works_on table includes emp_id, prj_id, and hours. The dependent table includes emp_id, dpd_name, sex, birth, and relationship.}
\end{figure}

\begin{figure}[!ht]
    \includegraphics[width=0.9\textwidth]{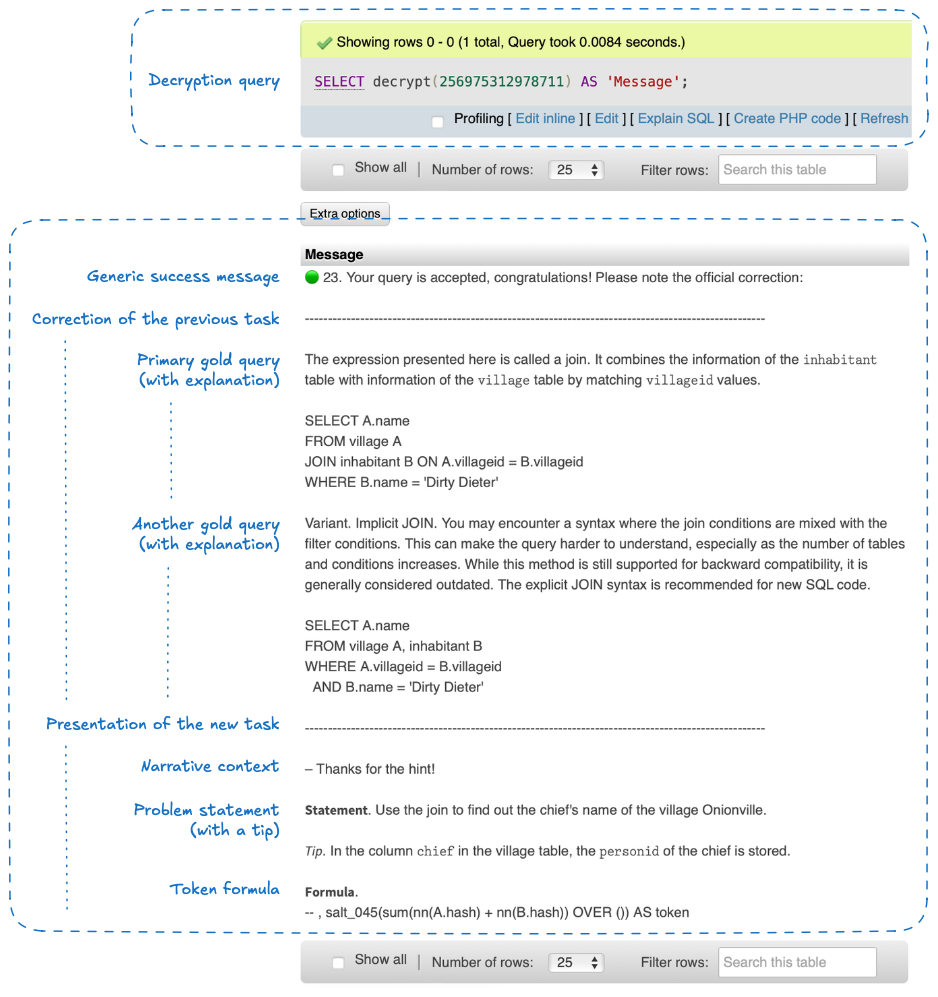}
    \caption{Accessing a feedback message of the MySQL version of SQLab Island under phpMyAdmin \cite{phpmyadmin}. The text is decrypted when the student submits a correct answer. It is a concatenation of a generic success message, the correction of the previous task (here, the primary gold query and another one), and the presentation of the new task (the narrative context, the problem statement, and the commented token formula). Although phpMyAdmin renders the tables in plain text, an automatic conversion of Markdown to Unicode mathematical characters allows for the display of bold, italic, and monospace texts.}
    \label{fig:pma}
    \Description{An annotated screenshot of a typical answer unlocked by SELECT decrypt(token). It consists of a generic success message, the correction of the previous task (the primary gold query and another gold query, both with an explanation), and the presentation of the new task (the narrative context, the problem statement with a tip, and the token formula).}
\end{figure}

\begin{table}[!ht]
    \tablecaption{A case of collision with \texttt{bit\_xor()}. The left [resp., right] query counts the number of dependents [resp., \emph{male} dependents] per employee.
    If the token formula implements $f_w(f_a())$ as \texttt{sum(bit\_xor())}, both queries have the same fingerprint (columns \texttt{token\_1}).
    Indeed, every employee happens to have either zero or an odd number of dependents, and either zero or one \emph{male} dependent. In both cases, the left join creates an odd number of occurrences of each \texttt{employee} row. Xoring a value with itself an odd number of times has no effect, resulting in the same token for both queries.
    The issue is resolved by implementing $f_w(f_a())$ as \texttt{bit\_xor(sum())} (columns \texttt{token\_2}), since summing three times a non-zero value is obviously different from summing it once.
    Of course, here, the natural choice is a two-dimensional formula, which would have avoided the issue altogether. \hfill\accmat}
    \begin{tabular}{ll}
\begin{lstlisting}[basicstyle=\tiny\ttfamily]
SELECT A.emp_id
    , count(B.emp_id) as count
    , salt_042(sum(bit_xor(nn(A.hash))) OVER ()) as token_1
    , salt_042(bit_xor(sum(nn(A.hash))) OVER ()) as token_2
FROM employee A
LEFT JOIN dependent B ON A.emp_id = B.emp_id
GROUP BY A.emp_id
\end{lstlisting}
&
\begin{lstlisting}[basicstyle=\tiny\ttfamily]
SELECT A.emp_id
     , count(B.emp_id) as count
     , salt_042(sum(bit_xor(nn(A.hash))) OVER ()) as token_1
     , salt_042(bit_xor(sum(nn(A.hash))) OVER ()) as token_2
FROM employee A
LEFT JOIN dependent B ON A.emp_id = B.emp_id §\textcolor{ACMOrange}{\texttt{AND B.sex = 'M'}}§
GROUP BY A.emp_id
\end{lstlisting}
\\
\\
\begin{minipage}{0.49\linewidth}\scriptsize
\begin{verbatim}
  |   emp_id  | count |     token_1     |     token_2     |
  +-----------+-------+-----------------+-----------------+
  | 123456789 |   3   | 276085499732552 | 278507996506236 |
  | 333445555 |   3   | 276085499732552 | 278507996506236 |
  | 453453453 |   0   | 276085499732552 | 278507996506236 |
  | 666884444 |   0   | 276085499732552 | 278507996506236 |
  | 888665555 |   0   | 276085499732552 | 278507996506236 |
  | 987654321 |   1   | 276085499732552 | 278507996506236 |
  | 987987987 |   0   | 276085499732552 | 278507996506236 |
  | 999887777 |   0   | 276085499732552 | 278507996506236 |
\end{verbatim}
\end{minipage}
&
\begin{minipage}{0.49\linewidth}\scriptsize
\begin{verbatim}
  |   emp_id  | count |     token_1     |     token_2     |
  +-----------+-------+-----------------+-----------------+
  | 123456789 |   1   | 276085499732552 | 280652170274914 |
  | 333445555 |   1   | 276085499732552 | 280652170274914 |
  | 453453453 |   0   | 276085499732552 | 280652170274914 |
  | 666884444 |   0   | 276085499732552 | 280652170274914 |
  | 888665555 |   0   | 276085499732552 | 280652170274914 |
  | 987654321 |   1   | 276085499732552 | 280652170274914 |
  | 987987987 |   0   | 276085499732552 | 280652170274914 |
  | 999887777 |   0   | 276085499732552 | 280652170274914 |
\end{verbatim}
\end{minipage}
\end{tabular}
    \label{tab:xor_sum_vs_sum_xor}
\end{table}

\begin{table}[!ht]
    \tablecaption{
        Six queries to find the employees who work 40 hours (adapted and extended from \cite{galindo2001}). All perform a table combination, an aggregation, and a filter in various orders, as shown from top to bottom in the corner of each cell.
        $J$ and $C$ denote a join and a correlation, respectively.
        $A_s$ and $A_v$ denote a scalar and a vector aggregate, respectively.
        $F_w$ and $F_h$ denote a filter in a WHERE and a HAVING clause, respectively.
        \\
        The six queries produce three distinct tokens, one for each row. Row A uses Formula~\ref{eqn:formulas:agg} in a two-dimensional setting. Rows B and C use Formula~\ref{eqn:formulas:basic} in a one- and two-dimensional settings, respectively.
        As soon as a derived table is involved (A2, B2, C1, C2), it must handle part of the token calculation and pass it to the outer query. Students can hardly be expected to craft such decompositions on their own, which places this type of answer outside the pedagogical scope of our model.
        \\
        For readability, the derived tables are expressed as common table expressions (CTE), and the formulas are shortened by omitting the salt function $f_s$, the coalescing function $f^*$, and the alias “\lstinline|AS token|”. \hfill\accmat
        }
    
\begin{tabular}{r|l|l|}
\multicolumn{1}{l}{}&
\multicolumn{1}{c}{1}&
\multicolumn{1}{c}{2}
\\ \cline{2-3}

\multicolumn{1}{l}{}
&
\multicolumn{1}{r}{$J$}
&
\multicolumn{1}{r}{$J$}
\\

&
\multicolumn{1}{r}{$A_v$}
&
\multicolumn{1}{r}{$A_v$}
\\

&
\multicolumn{1}{r}{$F_w$}
&
\multicolumn{1}{r}{$F_h$}
\\[-3\normalbaselineskip]

A
&
\begin{lstlisting}[basicstyle=\footnotesize\ttfamily]
SELECT emp_id
     , emp_name
     , bit_xor(sum(A.hash + B.hash)) OVER ()
FROM employee A
JOIN works_on B USING (emp_id)
GROUP BY A.emp_id
HAVING 40 = sum(hours)
\end{lstlisting}
&
\begin{lstlisting}[basicstyle=\footnotesize\ttfamily]
WITH derived_table AS (
  SELECT emp_id
       , emp_name
       , sum(hours) AS sum_hours
       , sum(A.hash) as A_hash
       , sum(B.hash) as B_hash
  FROM employee A
  JOIN works_on B USING (emp_id)
  GROUP BY A.emp_id
)
SELECT emp_id
     , emp_name
     , bit_xor(A_hash + B_hash) OVER ()
FROM derived_table
WHERE 40 = sum_hours
\end{lstlisting}

\\\cline{2-3}

\multicolumn{1}{l}{}

&
\multicolumn{1}{r}{$C$}
&
\multicolumn{1}{r}{$C$}
\\

&
\multicolumn{1}{r}{$A_s$}
&
\multicolumn{1}{r}{$A_v$}
\\

&
\multicolumn{1}{r}{$F_w$}
&
\multicolumn{1}{r}{$F_h$}
\\[-3\normalbaselineskip]

B
&
\begin{lstlisting}[basicstyle=\footnotesize\ttfamily]
SELECT emp_id
     , emp_name
     , sum(A.hash) OVER ()
FROM employee A
WHERE 40 = (
  SELECT sum(hours)
  FROM works_on B
  WHERE A.emp_id = B.emp_id
)
\end{lstlisting}
&
\begin{lstlisting}[basicstyle=\footnotesize\ttfamily]
WITH derived_table AS (
  SELECT
      ( SELECT emp_id
        FROM works_on B
        WHERE A.emp_id = B.emp_id
        GROUP BY emp_id
        HAVING 40 = sum(hours)
      ) AS emp_id
      , emp_name
      , hash
  FROM employee A
)
SELECT emp_id
     , emp_name
     , sum(A.hash) OVER ()
FROM derived_table A
WHERE emp_id IS NOT NULL
\end{lstlisting}
\\\cline{2-3}

\multicolumn{1}{l}{}
&
\multicolumn{1}{r}{$A_v$}
&
\multicolumn{1}{r}{$A_v$}
\\

&
\multicolumn{1}{r}{$F_h$}
&
\multicolumn{1}{r}{$J$}
\\

&
\multicolumn{1}{r}{$J$}
&
\multicolumn{1}{r}{$F_w$}
\\[-3\normalbaselineskip]

C
&
\begin{lstlisting}[basicstyle=\footnotesize\ttfamily]
WITH derived_table AS (
  SELECT emp_id
       , sum(hours) as sum_hours
       , sum(hash) as hash
  FROM works_on
  GROUP BY emp_id
  HAVING 40 = sum(hours)
)
SELECT emp_id
     , emp_name
     , sum(A.hash + B.hash) OVER ()
FROM employee A
JOIN derived_table B USING (emp_id)
\end{lstlisting}
&
\begin{lstlisting}[basicstyle=\footnotesize\ttfamily]
WITH derived_table AS (
  SELECT emp_id
       , sum(hours) as sum_hours
       , sum(hash) as hash
  FROM works_on
  GROUP BY emp_id
)
SELECT emp_id
     , emp_name
     , sum(A.hash + B.hash) OVER ()
FROM employee A
JOIN derived_table B USING (emp_id)
WHERE 40 = sum_hours
\end{lstlisting}

\\\cline{2-3}

\end{tabular}
    \label{tab:six_variations}
\end{table}

\subsection{DBMS requirements} \label{sec:dbms-requirements}
\begin{enumerate}
\item \emph{Window functions.} First introduced in Oracle 8i (1998) under the name \emph{analytic functions}, they were added to the ANSI/ISO SQL Standard in 2003. Since then, they have been implemented in most major DBMSs, including
IBM Db2 (2008), 
PostgreSQL 8.4 (2009), 
SQL Server (2012), 
MariaDB 10.2 (2017), 
MySQL 8.0.2 (2018), 
SQLite 3.25.0 (2018), 
among others \cite{scutaru2024}. SQLab rely on window functions to compute and broadcast the access token across an entire column of the result table. 

\item \emph{Hash function.} During the creation of the database, a hash of each row is calculated and stored in a separate column. These values are used as cross-table identifiers, and combined to form the tokens. Aside from this main purpose, the designer may rely on the availability of a hash function to make the students convert a non-numeric control value into a number, as part of a two-pass fingerprinting. Although SQLab uses \texttt{sha2()} for MySQL (built-in), \texttt{sha256()} for PostgreSQL (built-in), or \texttt{crypto\_sha256()} for SQLite (extension via \cite{sqlean}), even a non-cryptocraphic standard hash algorithm should be suitable.

\item \emph{Decryption function.} The feedback messages are encrypted by the Python program that build the SQLab database, but they must be decrypted at the DBMS level each time a student submit a token. This requires a decryption function, like \texttt{aes\_decrypt()} in MySQL or \texttt{pgp\_sym\_decrypt()} in PostgreSQL. Note that SQLean 0.27.0 \cite{sqlean}, the set of easy-to-install SQLite extensions we rely on, does \emph{not} provide cryptographic encryption/decryption functions. For now, we obfuscate the messages with a combination of hex-encoding and Brotli compression \cite{brotli}, and prefix the result with a SHA-256 hash of the access token. This is an area where our SQLite support could be improved.

\item \emph{User Defined Functions.} SQLab defines and exposes to the students
the coalescing function \texttt{nn()},
the \texttt{string\_hash()} function,
the \texttt{decrypt()} function,
as well as a great number of instances of a salt function \texttt{salt\_ddd()}.
All need to be implemented as UDFs in the DBMS itself, preferably in a way that does not require any additional installation or compilation step by the end user. As a short example, an implementation of the \texttt{decrypt()} function is shown in Fig.~\ref{fig:mysql-decrypt}.

\begin{figure}[!ht]
\begin{lstlisting}
  CREATE FUNCTION decrypt(token BIGINT UNSIGNED) RETURNS TEXT DETERMINISTIC
  BEGIN
      DECLARE CONTINUE HANDLER FOR SQLWARNING BEGIN END;
      DECLARE message TEXT;
      SELECT coalesce(
                 max(CONVERT(uncompress(aes_decrypt(msg, token)) USING utf8mb4)),
                 CONVERT("{preamble_default}" USING utf8mb4)
             ) INTO message
      FROM sqlab_msg;
      RETURN message;
  END;
\end{lstlisting}
\caption{A MySQL implementation of the \texttt{decrypt()} UDF.}
\label{fig:mysql-decrypt}
\Description{Text.}
\end{figure}

\item \emph{Triggers or generated columns.} In a preliminary version of SQLab, the hash values were calculated during the build process. As a consequence, only stateless games could be supported. This approach was abandoned in favor of triggers, which can update the hash values whenever a DML statement inserts or modifies rows. Note that the MySQL triggers suffer from a limitation which prevents them from including any \texttt{auto-increment} column in the hash calculation. An experimental solution with generated columns has been developed, but seems difficult to port to SQLite.

\end{enumerate}

\subsection{Proof of Theorem~\ref{thm:accuracy} on fingerprinting accuracy} \label{sec:proof-accuracy}

\begin{proof}
\emph{Direct, by contradiction.} Assume that $Q_1$ and $Q_2$ produce the same token, but their starred versions $Q_1^\mathtt{*}$ and $Q_2^\mathtt{*}$ result in different tables. Since the outer FROM clause has no derived table, all the tables it refers to are part of the primary database. Thus, every row constructed by this FROM must be a concatenation of either rows from these tables, or sequences of NULL values (in case of an outer join). Let us call them \emph{primary tuples} and \emph{null tuples}, respectively.
\begin{enumerate}
\item Suppose that the outer FROM clauses of $Q_1$ and $Q_2$ refer to the same list of tables.
\begin{enumerate}
    \item \label{itm:direct:1}
    Suppose without loss of generality that there exists a primary tuple of $Q_1^\mathtt{*}$ that is only constructed by the outer FROM clause of $Q_1$. Since $h$ is injective, its hash value is distinct from any other hash value. In particular, it does not appear in any row constructed by the outer FROM clause of $Q_2$. The formula of $Q_1$ takes this unique hash value as an input. Being a pipeline of injective functions, it necessarily produces a token distinct from that of $Q_2$: contradiction.
    \item \label{itm:direct:2}
    Suppose that all the primary tuples constructed by the outer FROM clause of $Q_1^\mathtt{*}$ are also constructed by the outer FROM clause of $Q_2^\mathtt{*}$, and conversely. The differing tuples must be null. The appearance of null tuples can only happen in the case of an outer join. Suppose without loss of generality that it is a \emph{left} outer join. Consider the corresponding non-matching tuples of the left side of the join. They must be non-null, i.e., primary. The fact that the set of primary tuples is the same for $Q_1^\mathtt{*}$ and $Q_2^\mathtt{*}$ has two consequences:
    \begin{enumerate}
        \item The token difference cannot come from these tuples. It must come from the null tuples. A null tuple has a NULL value in the \texttt{hash} column, but $f^*$ maps this NULL to a constant positive integer. So the contribution of the null tuples to the tokens depends only on their cardinality.
        \item There are as many null tuples for $Q_1$ as for $Q_2$. Hence, the token difference cannot come from these tuples either. Contradiction with the hypothesis.
    \end{enumerate}
\end{enumerate}
\item Suppose without loss of generality there exists a table $T_1$ only referred to in the outer FROM clause of $Q_1$.
\begin{enumerate}
    \item If at least one tuple of $T_1$ constructed by the outer FROM clause of $Q_1^\mathtt{*}$ is primary, cf. step \ref{itm:direct:1}.
    \item Suppose that all the tuples of $T_1$ constructed by the outer FROM clause of $Q_1^\mathtt{*}$ are null. This can only happen in the case of a useless outer join, i.e., an outer join having no matching rows in $T_1$. The reasoning of step \ref{itm:direct:2} implies that there exists a table $T_2$ only referred to in the outer FROM clause of $Q_2$, and having as many null tuples as $T_1$. We can indeed construct two queries that produce the same token but different tables by left-joining the same table $T$ to $T_1$ and $T_2$, respectively, on a condition that is always false. The result of $Q_1^\mathtt{*}$ and $Q_2^\mathtt{*}$ will then differ on the null columns of $T_1$ and $T_2$, respectively. However, the theorem states that these columns are ignored.
\end{enumerate}
\end{enumerate}

\emph{Converse.} This is the obvious direction. Let $Q_1$ and $Q_2$ be two queries whose starred versions $Q_1^\mathtt{*}$ and $Q_2^\mathtt{*}$ result in the same table $T$. Since $Q_1$ and $Q_2$ have the same deterministic formula, and take as input the same values in $T$, they must produce the same token.
\end{proof}

\subsection{Default handling of unintended formulas in three DBMSs} \label{sec:incompatibility_handling}

Assume that the students are tasked to deduplicate a column \texttt{col}. There are two methods to achieve this:
\begin{enumerate}
\item either with a DISTINCT keyword in a basic query, which requires Formula~\ref{eqn:formulas:basic};
\item or (questionably) with a GROUP BY clause in an aggregated query, which requires Formula~\ref{eqn:formulas:agg}.
\end{enumerate}

The students are not free to choose the formula, as it is provided in the statement. Consequently, there is a risk that they will use the wrong method. Here is how three DBMSs would handle the cases where the query is “incompatible” with the formula:

\begin{enumerate}
\item \lstinline|SELECT DISTINCT col, formula_|\texttt{\ref{eqn:formulas:agg}} \lstinline|FROM A| \hfill\accmat
\begin{description}
    \item[SQLite.] \emph{No error.}
    \item[MySQL.] \texttt{1140: In aggregated query without GROUP BY, expression \#1 of SELECT list contains nonaggregated column 'a.col'; this is incompatible with sql\_\allowbreak{}mode=only\_\allowbreak{}full\_\allowbreak{}group\_\allowbreak{}by.}
    \item[PostgreSQL.] \texttt{GroupingError: column "a.col" must appear in the GROUP BY clause or be used in an aggregate function.}
\end{description}

\item \lstinline|SELECT col, formula_|\texttt{\ref{eqn:formulas:basic}} \lstinline|FROM A GROUP BY col| \hfill\accmat
\begin{description}
    \item[SQLite.] \emph{No error.}
    \item[MySQL.] \texttt{1055: Expression \#2 of SELECT list is not in GROUP BY clause and contains nonaggregated column 'A.hash' which is not functionally dependent on columns in GROUP BY clause; this is incompatible with sql\_\allowbreak{}mode=only\_\allowbreak{}full\_\allowbreak{}group\_\allowbreak{}by.}
    \item[PostgreSQL.] \texttt{GroupingError: column "a.hash" must appear in the GROUP BY clause or be used in an aggregate function.}
\end{description}

\end{enumerate}

SQLite is quite permissive in these situations. In addition to yielding tokens, these tokens are distinct. This allows the designer to provide a hint regarding the recommended style.

Both MySQL and PostgreSQL are strict about the compatibility of queries with formulas. Unfortunately, the messages they produce in such cases are likely to be confusing for the students. Furthermore, the designer has no means of substituting a specific hint to these runtime errors. While it is uncommon to be able to employ both an aggregated and a non-aggregated query to achieve the same result, the issue is bound to arise when an aggregation is used where it is not expected, or vice versa. The instructor can explain that the formula should be commented during debugging, as it is only useful for calculating the final token. Consistent with this advice, SQLab always provides the formula in the form of a comment (see bottom of Fig.~\ref{fig:pma}).

\end{document}